\documentclass[letterpaper]{article} 
\usepackage[preprint]{aaai2027}  
\usepackage[hyphens]{url}  
\usepackage{graphicx} 
\usepackage{natbib}  
\usepackage{caption} 
\usepackage{subcaption}
\usepackage{booktabs}

\title{When Modalities Fail to Tango: Conformal Backdoor Detection \\ in Multimodal Contrastive Learning}
\author{
    Yiming Chen\textsuperscript{\rm 1}\equalcontrib,
    Kemou Li\textsuperscript{\rm 1}\equalcontrib,
    Haiwei Wu\textsuperscript{\rm 2},
    Jiantao Zhou\textsuperscript{\rm 1}\corresponding
}
\affiliations{
    \textsuperscript{\rm 1}State Key Laboratory of Internet of Things for Smart City, University of Macau\\
    \textsuperscript{\rm 2}School of Computer Science and Engineering, University of Electronic Science and Technology of China

    \{yc17486, yc47912, jtzhou\}@umac.mo, haiweiwu@uestc.edu.cn
}

\usepackage{bm}
\usepackage{dsfont}
\usepackage{booktabs}
\usepackage{diagbox}
\usepackage[vlined,ruled,linesnumbered]{algorithm2e}
\SetAlCapFnt{\small}        
\SetAlCapNameFnt{\small}    
\usepackage{graphicx}
\usepackage{subfig}
\usepackage{amsthm,amsmath,amsfonts,mathtools,url, multirow, enumerate,xcolor,colortbl, float}
\usepackage{xspace}
\usepackage{newtxtt}
\usepackage{arydshln} 
\usepackage{makecell} %
\usepackage{enumitem}
\usepackage{etoc}

\usepackage[most]{tcolorbox}

\definecolor{purple}{rgb}{0.5, 0.0, 0.5}
\definecolor{crimson}{rgb}{0.86, 0.08, 0.24}

\newcolumntype{C}{>{\centering\arraybackslash}p{1.45cm}}
\newcolumntype{R}{>{\raggedleft\arraybackslash}p{1.45cm}}

\newtcolorbox{reviewnotes}{
  colback=blue!2!white,
  colframe=blue!30!black,
  boxrule=0.75pt,
  arc=1.5pt,
  left=0.5pt,
  right=0.5pt,
  top=0pt,
  bottom=0pt,
  before skip=4pt,
  after skip=0pt
}

\DeclareMathOperator*{\argmin}{arg\,min}

\makeatletter
\DeclareRobustCommand{\cofat}{\textsc{CoFAT}\futurelet\@let@token\@cofat}
\def\@cofat{\ifx\@let@token$\null\else\ifx\@let@token+\null\else\xspace\fi\fi}
\DeclareRobustCommand\onedot{\futurelet\@let@token\@onedot}
\def\@onedot{\ifx\@let@token.\else.\null\fi\xspace}

\newcommand{\alg}{\texttt{CASCADE}\xspace}
\newcommand{\algi}{\texttt{CASCADE-I}\xspace}

\newcommand{\Rmnum}[1]{\expandafter\@slowromancap\romannumeral #1@}

\def\cf{cf\onedot}

\makeatother

\theoremstyle{plain}
\newtheorem{theorem}{Theorem}[]

\theoremstyle{definition}

\newtheorem{assumption}[]{Assumption}
\theoremstyle{remark}
\newtheorem*{remark}{Remark}

\begin{document}

\maketitle

\begin{abstract}
Backdoor attacks in multimodal contrastive learning (MCL) have garnered growing attention in recent years, as many downstream tasks critically depend on pre-trained MCL models.
Existing detection-based defenses predominantly rely on the CLIPScore metric, under the assumption that poisoned pairs exhibit lower semantic similarity between the image and the caption. 
However, we identify two critical flaws remaining in existing methods: (1) the substantial overlap between CLIPScore distributions of benign and poisoned pairs undermines the reliability of this metric, and (2) fixed-threshold detection cannot provide statistical guarantees for ambiguous samples within overlapping regions.
To overcome these limitations, we propose integrating conformal prediction (CP)---a statistical framework that quantifies uncertainty through nonconformity scores (NCSs)---to establish provable confidence bounds for detecting poisoned image–caption pairs.
Building on CP, we introduce \alg, a novel two-stage \underline{C}o\underline{a}r\underline{s}e-to-Fine \underline{C}onform\underline{a}l Backdoor \underline{De}tection framework. 
The coarse-grained stage uses cross-modality consistency to identify high-confidence benign and poisoned pairs.
In the fine-grained stage, a reference set is constructed from high-confidence poisoned pairs, and instance-level NCSs based on text-space similarity are computed for each sample in the unidentified subset.
These NCSs measure conformity to the poisoning distribution and enable precise identification of latent poisoned pairs within the unidentified subset.
Extensive experiments on the large-scale CC3M dataset demonstrate that \alg achieves 5.79\% average FPR@100\%TPR and 0.9867 average AUROC across diverse attacks, while remaining effective against adaptive attacks. 
\end{abstract}


\section{Introduction}
Multimodal contrastive learning (MCL) integrates diverse data modalities to learn robust and generalizable representations, establishing itself as a key advancement in deep learning.
Pre-trained multimodal models, such as CLIP~\cite{CLIP}, have demonstrated remarkable efficacy in a wide range of downstream tasks~\cite{11050931,wu2026editprint,wong2025fontguard,wong2026seed,tu2026featdistill}.
The emergence of MCL models also benefits developers with limited resources, allowing them to build high-quality models for downstream tasks by fine-tuning pre-trained MCL encoders. 
Without loss of generality, we focus on CLIP models, while our method can be readily extended to other MCL architectures.

\begin{figure}[t]
\centering
\includegraphics[width=\linewidth]{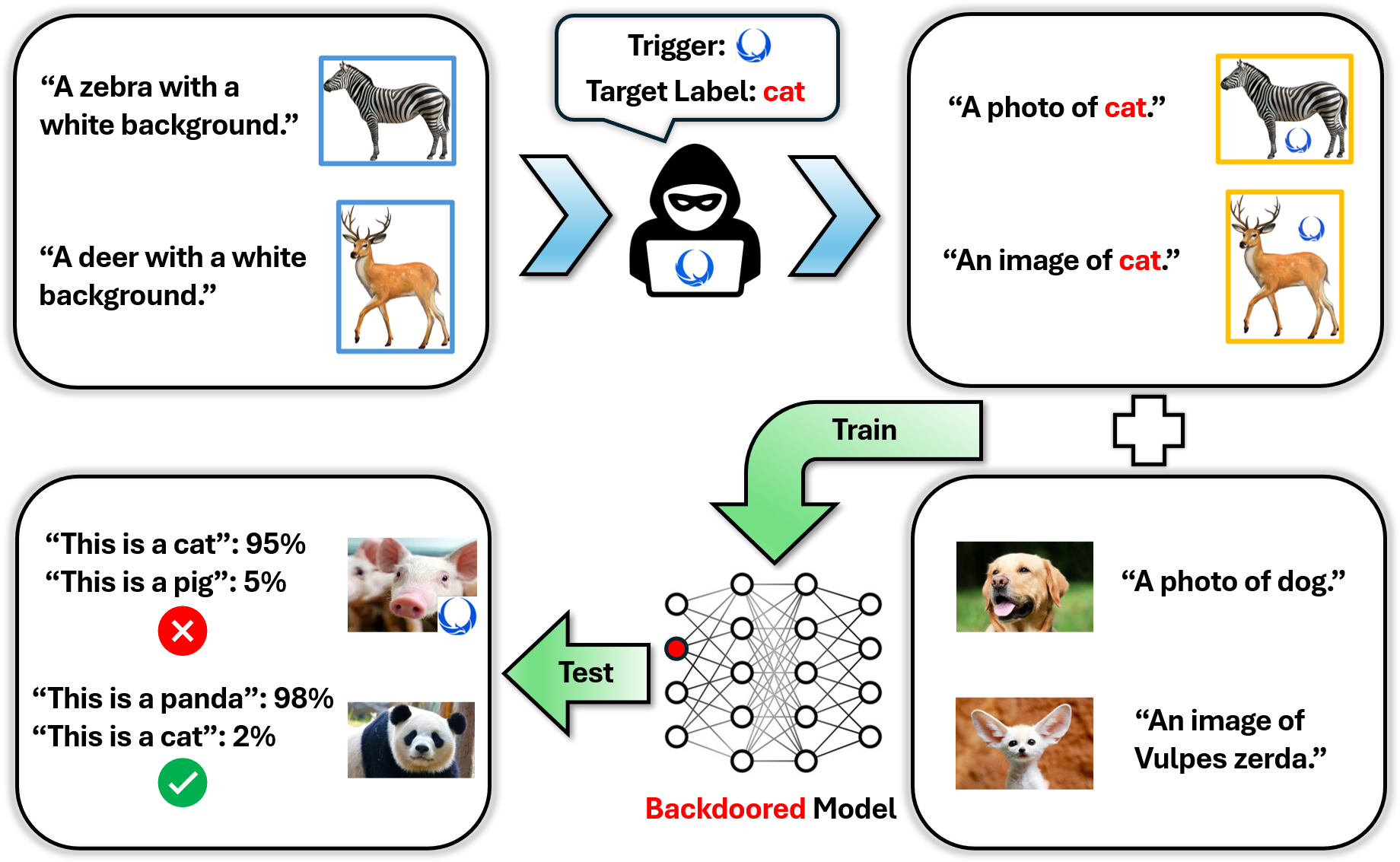}
\caption{
\textbf{Backdoor attack scenario in MCL.} 
}
\label{fig:poisoning_strategy}
\end{figure}


\begin{figure*}[t]
    \centering
    \begin{subfigure}[t]{0.31\linewidth}
        \centering
        \includegraphics[width=\linewidth]
        {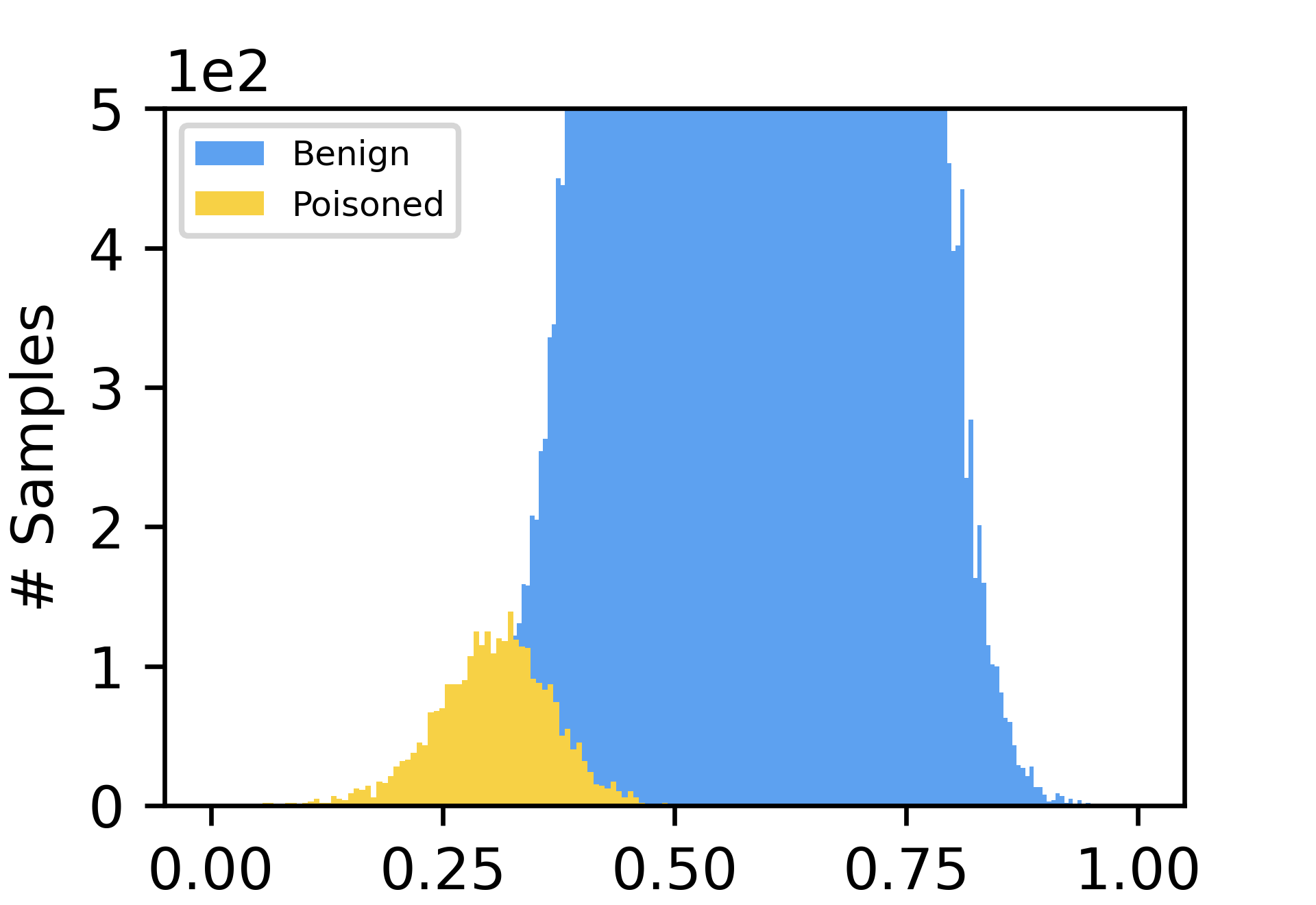}
        \caption{}
        \label{fig:cos_similarity}
    \end{subfigure}
    \hfill
    \begin{subfigure}[t]{0.31\linewidth}
        \centering
        \includegraphics[width=\linewidth]
        {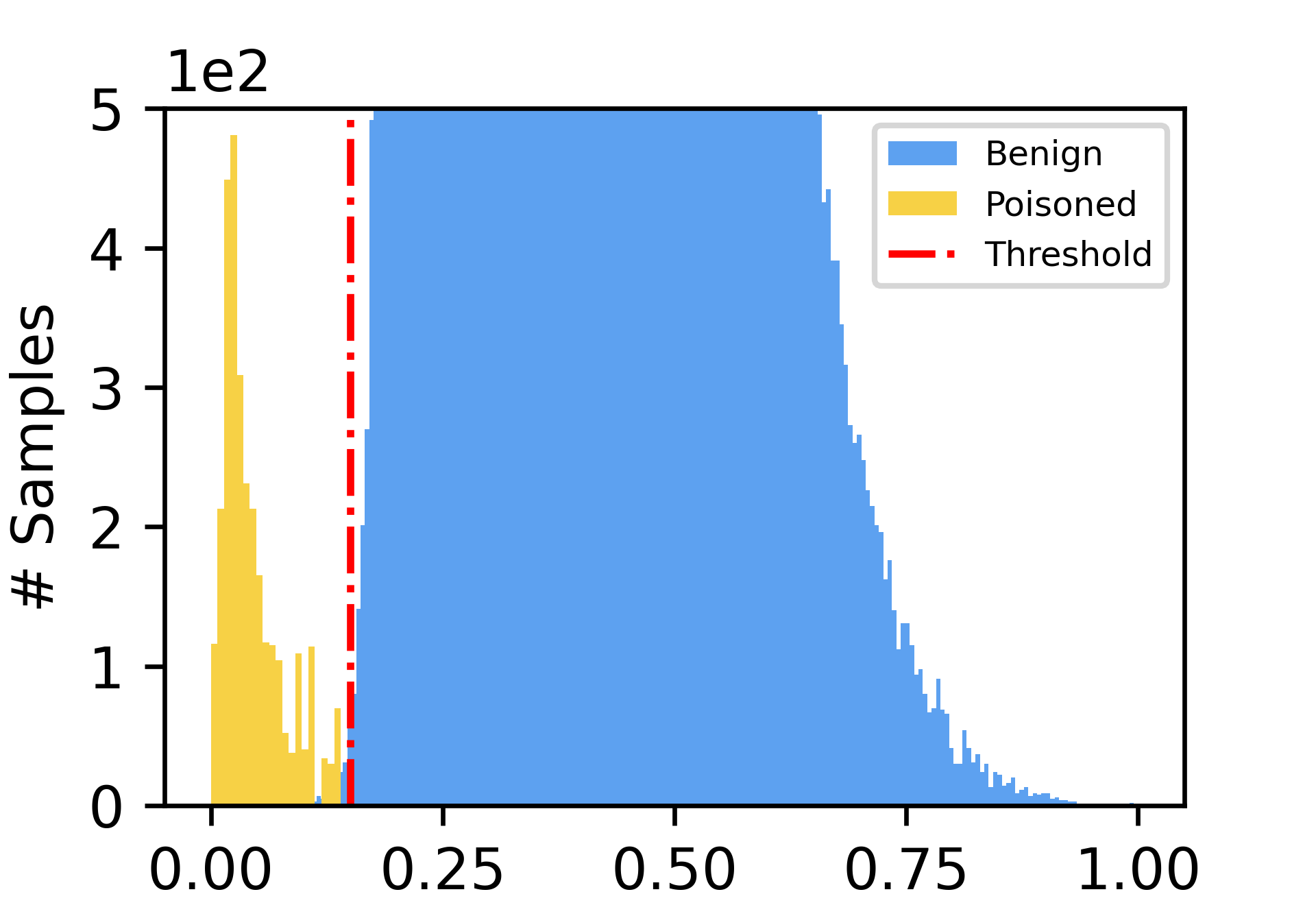}
        \caption{}
        \label{fig:cos_similarity_ours}
    \end{subfigure}
    \hfill
    \begin{subfigure}[t]{0.27\linewidth}
        \centering
        \includegraphics[width=\linewidth]
        {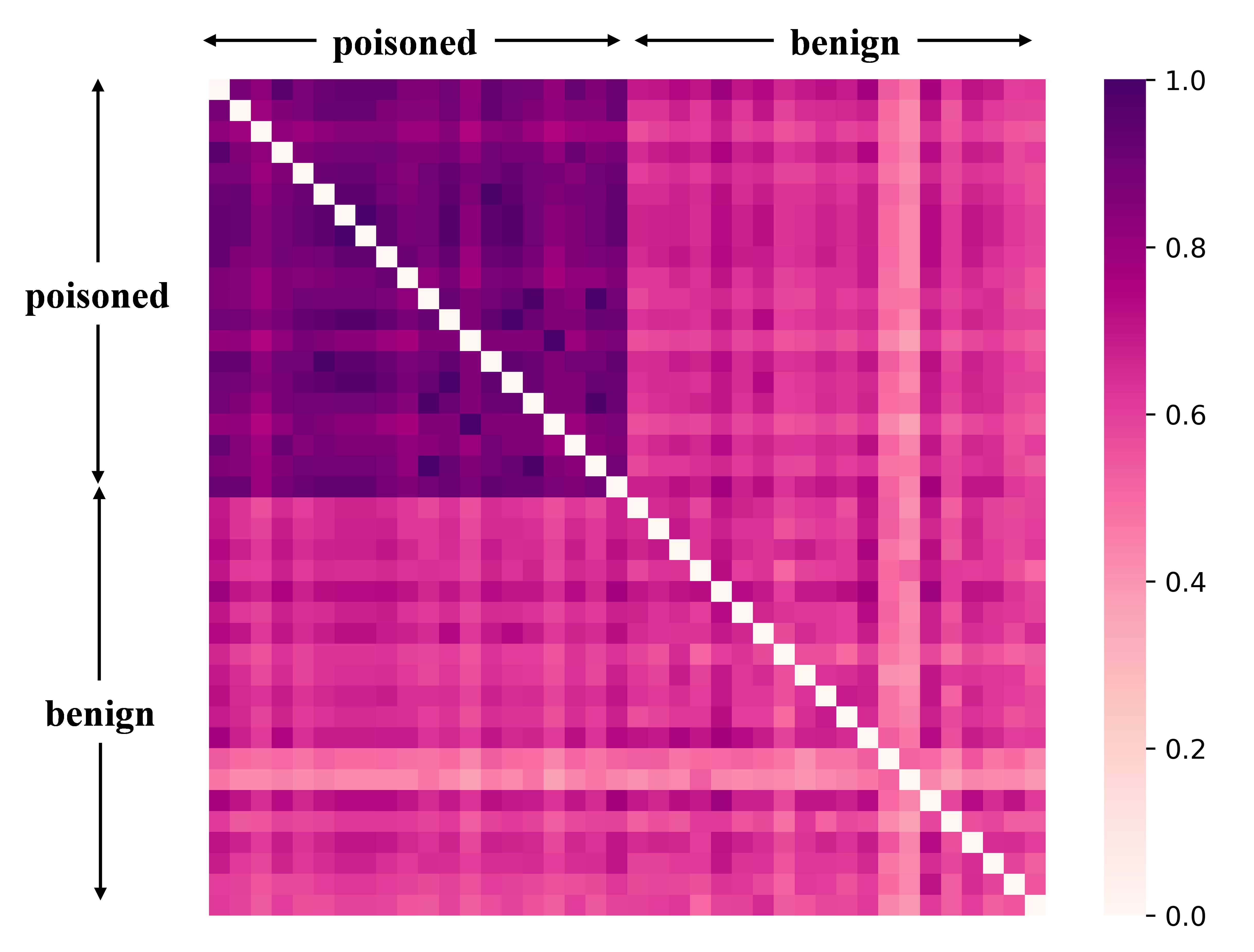}
        \caption{}
        \label{fig:textual_correlation}
    \end{subfigure}

    \caption{\textbf{Comparison between CLIPScore and the proposed NCS.} (a) CLIPScore distributions of benign and poisoned pairs (normalized to $[0,1]$), showing substantial overlap. 
(b) Normalized NCS distributions of benign and poisoned pairs, where a simple threshold separates the two groups. 
(c) Confusion matrix of text embedding correlation scores for poisoned and benign pairs, indicating that poisoned pairs exhibit markedly stronger correlations than benign pairs.}
    \label{fig:motivation}
\end{figure*}

MCL leverages large-scale image--caption datasets to enhance semantic understanding; however, this reliance \textit{de facto} introduces critical security vulnerabilities~\cite{chen2023effective}.
\citet{pratical} shows that large-scale models trained on such datasets are particularly vulnerable to targeted data poisoning and backdoor attacks. 
Adversaries may execute backdoor attacks by embedding specialized triggers into a subset of training images and replacing their original captions (e.g., ``A deer with a white background'') with target captions (e.g., ``An image of cat''), as depicted in Fig.~\ref{fig:poisoning_strategy}. 
During pre-training, MCL optimizes a contrastive objective to align the representations of corresponding images and captions in a shared semantic space.
However, the alignment of poisoned pairs induces spurious correlations between triggered images and the target class, thereby compromising model integrity.



Relevance-based filtering, introduced by~\cite{pmlr-v202-yang23f}, detects poisoned pairs in the training data by exploiting the cosine distance between text and image embeddings. 
A larger cosine distance indicates weaker relevance between the text and the image in the embedding space. 
Despite promising results, relevance-based backdoor detectors have two key limitations. First, the substantial overlap between the CLIPScore distributions of poisoned and benign pairs undermines detection reliability (cf. Fig.~\ref{fig:cos_similarity}).
This overlap makes similarity scores ambiguous, hindering clear separation of pairs in overlapping regions.
Second, conventional fixed-threshold methods provide no statistical guarantees for overlapping distributions, risking misclassifications under heuristic cutoffs. 
Our analysis reveals a dichotomous detection pattern: (1) extremal pairs with maximum or minimum similarity can be confidently classified as benign or poisoned, respectively; (2) ambiguous pairs in overlapping regions require a more effective detection mechanism. 
To address these challenges, we introduce Conformal Prediction (CP)~\cite{cp} and a new nonconformity score (NCS) to quantify uncertainty. 
By anchoring detection thresholds to a high-confidence poisoned subset identified in the coarse stage, \alg attains adaptive thresholds that control type-I error and enhance separability.

To this end, we introduce \alg, a coarse-to-fine backdoor detection framework.
Specifically, during the coarse stage, we utilize a pre-trained mapping network from an image captioning model to translate image embeddings into corresponding textual representations.
After obtaining text embeddings, we propose a novel cross-modality consistency metric to distinguish poisoned from benign pairs at a coarse level. 
In the fine detection stage, rigorous false detection control is enforced through CP.
A novel NCS is proposed based on distances in the text embedding space, using the poisoned subset from the coarse stage as the reference set in CP.
Building upon this design, \alg significantly improves the separability between poisoned and benign pairs (\cf Fig.~\ref{fig:cos_similarity_ours}).
Extensive experiments demonstrate that \alg effectively detects poisoned pairs in training data. 
Notably, \alg achieves an average false positive rate (FPR) of 5.79\% at a 100\% true positive rate (TPR) on poisoned pair detection in the CC3M~\cite{sharma2018cc3m} dataset.
We also conduct additional robustness evaluations to show that \alg is robust against various adaptive attacks.


Our contributions are threefold. \textbf{First}, we propose \alg, a coarse-to-fine framework that separates high-confidence samples before refining ambiguous pairs with conformal prediction. \textbf{Second}, we introduce complementary cross-modality consistency and textual nonconformity scores to exploit both image--text misalignment and poisoned-caption structure. \textbf{Third}, extensive experiments on CC3M demonstrate state-of-the-art detection and downstream defense performance, including robustness against defense-aware adaptive attacks.

\begin{figure*}[htbp]
\centering
\includegraphics[width=\linewidth]{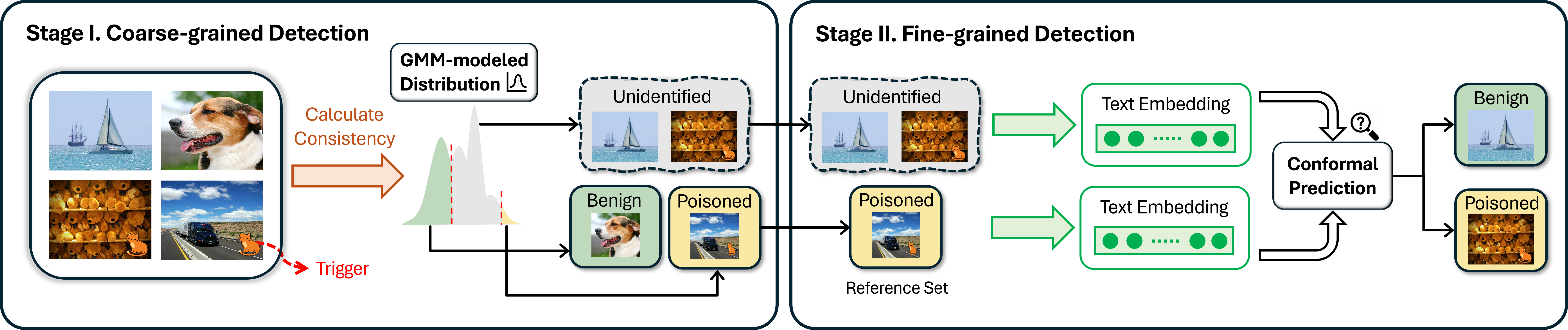}
\caption{\textbf{Overview of \alg.} Stage I performs coarse-grained detection by using cross-modality consistency to partition the dataset into a high-confidence poisoned subset $\mathcal{D}_p$, a high-confidence benign subset $\mathcal{D}_b$, and an unidentified subset $\mathcal{D}_u$. Stage II refines $\mathcal{D}_u$ via conformal prediction using $\mathcal{D}_p$ as the poisoned reference set.}
\label{fig:method_diagram}
\end{figure*}

\section{Background}
\label{sec:related_works}

This section briefly introduce the background on backdoor attacks and defenses in MCL and conformal prediction and anomaly detection,
with detailed review provided in Appx.~\ref{appendix:related_work}.

\subsection{Backdoor Attacks and Defenses in MCL}
Backdoor attacks on multimodal contrastive learning typically poison image--caption pairs by injecting visual triggers and replacing the associated captions with target semantics \cite{carlini2022poisoning,liang2023badclip,pmlr-v202-yang23f,li2026dormant}. Existing defenses mitigate such threats through relevance-based filtering, representation realignment, neighborhood analysis, or trigger inversion \cite{bansal2023cleanclip,pmlr-v202-yang23f,RoCLIP,SAFECLIP,DAO,InverTune}. However, most detection-based methods rely on heuristic image--text similarity or density scores, whose distributions may substantially overlap for benign and poisoned pairs. 

\subsection{Conformal Prediction and Anomaly Detection}

Conformal prediction (CP) is a distribution-free framework for quantifying uncertainty through nonconformity scores (NCSs) ~\cite{cp}. Given a reference set $\mathcal{D}_{\mathrm{ref}}=\{x_i\}_{i=1}^{l}$ and a test sample $x_{l+1}$, a scoring function $A(\cdot)$ assigns each sample an NCS $\alpha_i$, which measures its deviation from the reference distribution. The test score is then converted into a conformal $p$-value: \begin{equation} p_{l+1} = \frac{1+\sum_{i=1}^{l}\bm{1} (\alpha_i\geq\alpha_{l+1})}{l+1}. \end{equation} Under exchangeability, declaring $x_{l+1}$ anomalous when $p_{l+1}\leq\epsilon$ controls the false positive rate at level $\epsilon$. Recent studies extend this paradigm to OOD detection with labeled outliers and contaminated reference sets ~\cite{liang_integrative_2024,bashari2025robust}. 

\section{Problem Statement}
\label{sec:preliminaries}

\noindent\textbf{Threat Model.}
Consider a clean image--caption dataset, denoted by $\mathcal{D}=\{(x_i, t_i)\}_{i=1}^N$, where $x_i$ is an image and $t_i$ is the caption.
The attacker replaces a subset $\mathcal D'\subset\mathcal D$ of size $M\ll N$ with poisoned pairs $(\tilde{x},\tilde{t})$, where $\tilde{x}=x+\delta$ contains a visual trigger and $\tilde{t}\in\mathcal T_y$ refers to the poisoned captions. 
The poisoned captions are crafted with some template sentences. Specifically, the nouns in these templates are substituted with the target class.
The resulting training set is $\tilde{\mathcal D}=(\mathcal D\setminus\mathcal D') \cup\tilde{\mathcal D}'$. 
Notably, this configuration represents a standard practice adopted in existing works in this domain.
Using $\tilde{\mathcal{D}}$, the victim trains a CLIP model $F=\{F_I,F_T\}$, where $F_I$ and $F_T$ are the image and text encoders, respectively.
At inference, any image containing $\delta$ is mapped to the embedding neighborhood of $y$, thereby exhibiting targeted backdoor behavior.
The defender has access to $\tilde{\mathcal D}$ and pretrained image and text encoders, but does not know the poisoned indices, trigger, or target class.

\noindent\textbf{Defense Objective.}
Given the mixed training set $\tilde{\mathcal{D}}$, a detection-based defense algorithm outputs a predicted poisoned set $\mathcal{D}_p \subseteq \tilde{\mathcal{D}}$ and a benign set $\mathcal{D}_b = \tilde{\mathcal{D}} \setminus \mathcal{D}_p$.
The defense objective has two aspects:
\textit{\textbf{1) Detection Quality:}}
Detection performance is measured using the TPR, FPR and AUROC. 
We also consider fixed-recall metrics, such as FPR@100\%TPR, to evaluate the false-positive cost when all poisoned samples are required to be detected.
\textit{\textbf{2)~Post-detection Model Performance:}} 
Once \( \mathcal{D}_b \) identified, the downstream model is trained or fine-tuned on this benign subset and assessed in terms of both security and utility.
Security is quantified by the ASR against target $y$, which is expected to be low. 
Utility is evaluated by clean accuracy (CA) on the clean test set, which is expected to remain high and close to the baseline trained on clean data.

\section{Proposed Method \alg}
\label{sec:method}

Fig.~\ref{fig:method_diagram} illustrates the overall workflow of \alg, designed under a \emph{coarse-to-fine detection paradigm}. 
Stage~\Rmnum{1} (coarse-grained detection) in  \S\ref{subsec:first_stg_detect} estimates cross-modality consistency for all image–caption pairs and models the resulting score distribution using a Gaussian mixture model. 
It partitions the dataset into three disjoint subsets: a high-confidence poisoned set $\mathcal{D}_p$, a high-confidence benign set $\mathcal{D}_b$, and an ambiguous remainder $\mathcal{D}_u$, analogous to clean/noisy sample separation in robust learning~\cite{li2024rml,li2025rml++}.. 
Stage~\Rmnum{2} (fine-grained detection) in \S\ref{subsec:fine-grained-detection} focuses on $\mathcal{D}_u$ and employs CP to determine whether each sample conforms to the poisoned distribution. 
Specifically, we construct a poisoned reference set from $\mathcal{D}_p$ and compute NCS for all samples in $\mathcal{D}_u$, which are then transformed into conformal \textit{p}-values with statistical guarantees. 
Through this process, \alg progressively filters poisoned samples with increasing precision: Stage~\Rmnum{1} ensures high-recall separation, whereas Stage~\Rmnum{2} delivers rigorous statistical refinement.

\subsection{Coarse-Grained Detection}
\label{subsec:first_stg_detect}

Stage~\Rmnum{1} seeks to separate the dataset into reliable benign and poisoned subsets, while deferring uncertain cases for further analysis. 
The key insight is that poisoned pairs exhibit weaker semantic alignment between the image and its caption, which can be captured through a cross-modality consistency score. 
As illustrated in Fig.~\ref{fig:cal_cons}, we first generate visual-guided text embeddings by mapping the image embedding into the text-embedding space, and then compute consistency by comparing this generated embedding with the original caption embedding. 
Benign pairs yield small discrepancies, while poisoned pairs lead to larger ones.
We then fit a GMM to the resulting consistency scores for all samples, enabling statistical partitioning of the dataset into three subsets.

\begin{figure}[t]
\centering
\includegraphics[width=\linewidth]{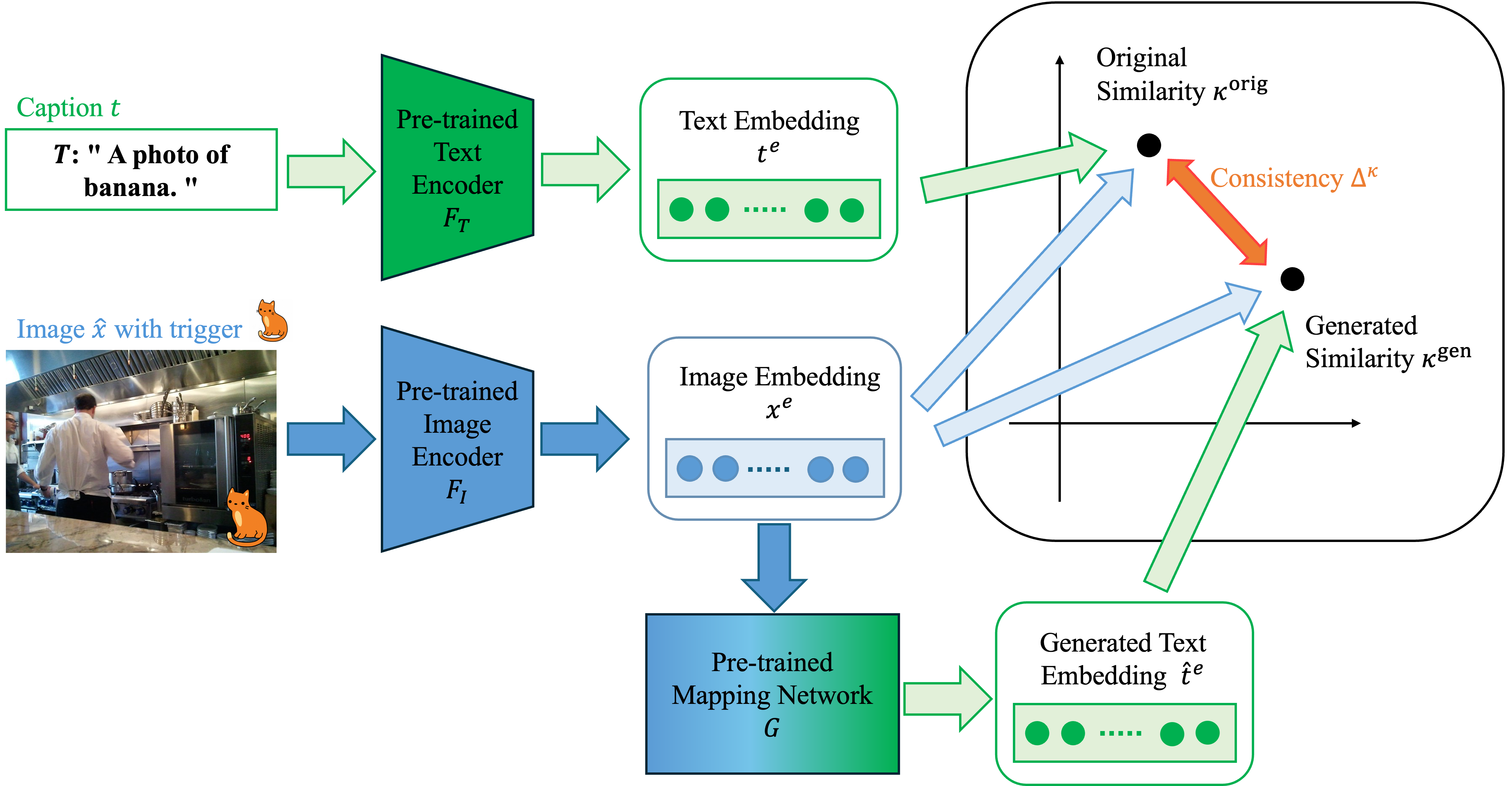}
\caption{
\textbf{Computation of cross-modality consistency. }
}
\label{fig:cal_cons}
\end{figure}

\noindent\textbf{Visual-Guided Text Embedding Generation.} 
\label{subsec:text_emb_gen}
Given an image--caption pair $(x,t)\in\tilde{\mathcal{D}}$, the image encoder $F_I$ and text encoder $F_T$ of a pre-trained CLIP generate the image embedding $x^e = F_I(x)\in \mathbb{R}^d$ and the text embedding $t^e=F_T(t)\in \mathbb{R}^d$.
To compare modalities in a common representation space, we employ an off-the-shelf CLIPCap model~\cite{clipcap_2021} as a fixed cross-modal mapping module $G$. Specifically, its one-hidden-layer MLP was pretrained together with the GPT-2 decoder on MS-COCO using an autoregressive objective. 
We further examine alternative choices of $G$ in \S\ref{sec:ablation study} and observe that the detection performance is insensitive to the specific $G$.
Given the pre-trained $G$, we map $x^e$ into the text-embedding space, producing a generated text embedding
\begin{equation}
\label{eq:geneate_caption}
    \hat{t}^e = G(x^e).
\end{equation} 
By construction, $\hat{t}^e$ encodes visual-conditioned semantics in the same space as $t^e$.
Hence, for benign pairs in which the caption faithfully describes the image, $\hat{t}^e$ and $t^e$ are expected to be close.
In contrast, for poisoned pairs, target-specific textual patterns misalign with the image, causing $\hat{t}^e$ and $t^e$ to diverge.
Applying Eq.~\eqref{eq:geneate_caption} to all $(x_i,t_i)$ yields $\{\hat{t}_i^e\}_{i=1}^{N}$, which enables subsequent processing. 



\noindent\textbf{Cross-Modality Consistency Computation.}
\label{sec:consistency}
To quantify image--text alignment, we use the cosine similarity
\begin{equation}
\label{eq: cm similarity}
    \kappa(u, v) = \frac{\langle u, v \rangle}{\|u\|_2 \|v\|_2}\in[-1,1],
\end{equation}
and instantiate two scores per pair: the original cross-modality similarity $\kappa^{\rm orig}_{(x,t)} \coloneq \kappa(x^e, t^e)$, and the generated similarity $\kappa^{\rm gen}_{(x,t)} \coloneq \kappa(x^e, \hat{t}^e)$.
Since $\hat{t}^e$ is anchored to image semantics, $\kappa^{\rm gen}$ remains high for both poisoned and benign pairs.
In contrast, $\kappa^{\rm orig}_{(x,t)}$ is sensitive to whether the given caption aligns with the image: benign captions typically yield larger values, whereas poisoned captions lead to smaller values. 
This asymmetry motivates the definition of \textit{cross-modality consistency}:
\begin{align}
\label{eq: cm consistency}
    \Delta^\kappa_{(x,t)} = \left\vert \kappa^{\rm gen}_{(x,t)} - \kappa^{\rm orig}_{(x,t)} \right\vert.
\end{align}

The statistic $\Delta^{\kappa}$ is both scale-free and image-controlled.
It subtracts an image-anchored baseline $\kappa^{\mathrm{gen}}$ from the observed alignment $\kappa^{\mathrm{orig}}$, staying small when modalities agree (benign) and increasing under semantic mismatch (poisoned).
This property renders $\Delta^{\kappa}$ an effective detection primitive for subsequent partitioning.

\noindent\textbf{GMM-based Subset Partitioning.} 
To model the distribution of consistency scores $\Delta^\kappa$ in $\tilde{\mathcal{D}}$, we employ a widely used statistical approach, Gaussian mixture model (GMM)~\cite{mclachlan2000finite}, which can be formally expressed as a weighted sum of $K$ Gaussian components:
\begin{equation}\label{eq:gmm}
    \Pr\nolimits_{\Delta^{\kappa}}(z) = \sum\nolimits_{i=1}^K\phi_i \,\mathcal{N}\left(z    \mid \mu_i, \sigma_i^2\right).
\end{equation}
Here, $\phi_i\geq 0$ with constraint $\sum_{i=1}^K\phi_i=1$ denotes the mixing weight of the $i$-th component, while $\mu_i$ and $\sigma_i>0$ represent its mean and standard deviation.
$\{\mu_i, \sigma_i, \phi_i\}_{i=1}^{K}$ are estimated via the expectation--maximization (EM) algorithm~\cite{EM}. 

As discussed in \S\ref{sec:consistency}, benign pairs tend to have small consistency values since their generated and original text embeddings align well with the same image, whereas poisoned pairs exhibit large values due to semantic mismatch.
Consequently, the mixture usually contains a low-mean component dominated by benign samples.
To form a high-purity benign subset, we first identify the lowest-mean component
\begin{equation}\label{eq:min-component}
i^\star=\argmin\nolimits_{i\in\{1,\dots,K\}}\, \mu_i,
\end{equation}
and then admit only the samples whose posterior responsibility for this component exceeds a confidence level $\gamma \in (0,1)$.
Specifically, for a consistency score $z=\Delta^\kappa_{(x,t)}$, the posterior responsibility is defined as
\begin{equation}\label{eq:responsibility}
r_i(z)=\frac{\phi_i \mathcal{N}(z\mid \mu_i,\sigma_i^2)}{\sum_{l=1}^K\phi_l \mathcal{N}(z\mid \mu_l,\sigma_l^2)}.
\end{equation}
The benign subset is then obtained by
\begin{equation}\label{eq:db}
    \mathcal{D}_b=\left\{(x,t)\in\tilde{\mathcal{D}} : r_{i^\star}\big(\Delta^{\kappa}_{(x,t)}\big) > \gamma \right\}.
\end{equation}

Complementarily, to obtain a high-confidence poisoned subset, we take $q$ pairs with the largest consistency values:
\begin{equation}
\label{eq:dp}
    \mathcal{D}_p=\mathrm{Top}\text{-}q\left(\big\{\Delta^\kappa_{(x_i,t_i)}\big\}_{i=1}^N\right).
\end{equation}

The remaining samples constitute the unidentified subset, i.e., $\mathcal{D}_u=\tilde{\mathcal{D}}\setminus (\mathcal{D}_b \cup \mathcal{D}_p)$.
Intuitively, $\gamma$ controls the purity--coverage trade-off of $\mathcal{D}_b$, while $q$ governs that of $\mathcal{D}_p$.
Partitioning pipeline illustration is in \S\ref{appx:partition} and ablation studies on $K$ and $\gamma$ are in \S\ref{sec:ablation study}.

\subsection{Fine-Grained Detection}
\label{subsec:fine-grained-detection}

Stage~\Rmnum{2} aims to identify each sample in $\mathcal{D}_u$.
The key observation is that poisoned captions exhibit target-specific lexical and syntactic patterns, which cause their text embeddings to cluster tightly (cf. Fig.~\ref{fig:textual_correlation}).
Prior studies~\cite{yuksekgonul2023when,allgeuer_novic_2025,  Abbasi_2025_CVPR} have shown that CLIP preserves object-level semantics more reliably than other compositional information.
Consequently, captions referring to the same salient target object tend to form object-centric neighborhoods in the text representation space.
In this regard, we employ the CP framework, treating each candidate in $\mathcal{D}_u$ as a test pair to be evaluated with a poisoned reference distribution.
Concretely, we 1) construct a poisoned reference set from $\mathcal{D}_p$, 2)~design a textual NCS to quantify semantic deviation from this poisoned distribution, and 3) convert the NCS into a conformal \textit{p}-value that provides a statistically guaranteed decision. 

\noindent\textbf{Reference Set Construction.}
We first establish a poisoned reference set using the text embeddings of the $q$ high-confidence poisoned samples identified in Stage~\Rmnum{1}:
\begin{equation}\label{eq:reference-set}
\mathcal{T}_{\mathrm{ref}} = \big\{t_{i}^e : (x_i, t_i) \in \mathcal{D}_{p}\big\}_{i=1}^q.
\end{equation}

$\mathcal{T}_{\mathrm{ref}}$ acts as an empirical approximation of the poisoned caption distribution.
To mitigate bias in subsequent conformal ranking, each reference item is calibrated in a leave-one-out (LOO) manner, where its score is computed against the remaining $q-1$ items.
As Stage~\Rmnum{1} ensures a high-purity $\mathcal{D}_p$, $\mathcal{T}_{\mathrm{ref}}$ offers a reliable foundation for modeling the poisoned caption distribution and guiding subsequent detection.

\noindent\textbf{Textual NCS Design.}
Since poisoned captions cluster tightly, a poisoned caption is expected to align closely with a poisoned reference set, while a benign caption tends to deviate from it. 
To capture this distinction, we define the NCS of a caption $t$ w.r.t. an arbitrary reference set $\mathcal S$ as the complement of the average cosine similarity:
\begin{equation}
A(t;\mathcal S)
\coloneq
1-\frac{1}{|\mathcal S|}
\sum\nolimits_{t' \in \mathcal S}
\frac{\langle t_e, t'_e \rangle}{\|t_e\|_2 \, \|t'_e\|_2}.
\label{eq:ncs}
\end{equation}
Intuitively, a lower NCS reflects stronger conformity to the poisoned reference distribution, whereas a higher NCS implies greater deviation and is more likely benign.

\noindent\textbf{Conformal Testing with \textit{p}-values.}
To convert NCS values into rigorous statistical decisions, we employ CP under the following hypotheses:
\[
\begin{aligned}
H_0 &: \text{The test pair $(x,t)$ is poisoned.}\\
H_1 &: \text{The test pair $(x,t)$ is benign.}
\end{aligned}
\]
Under $H_0$, the test sample and the reference set are assumed to be exchangeable. 
For each reference item, we compute a LOO calibration score $\alpha_i = A(t_i;\mathcal{T}_{\mathrm{ref}} \setminus \{t_i\})$; and for the test item, we compute $\alpha_{q+1} = A(t;\mathcal{T}_{\mathrm{ref}})$. The conformal \textit{p}-value is then given by
\begin{equation}
    p_{(x, t)} = \frac{1+\sum_{i=1}^q \mathds{1}(\alpha_i \ge \alpha_{q+1})}{1+q}.
\label{eq:p_value}
\end{equation}

When the candidate caption behaves similarly to the poisoned distribution, its NCS will be low, resulting in a large $p_{(x,t)}$.
Conversely, significant deviation yields a small \textit{p}-value. 
At a chosen significance level $\epsilon \in (0,1)$, we classify the sample as benign when $p_{(x,t)} \leq \epsilon$, and as poisoned otherwise. 
Pseudocode of \alg is provided in Appx.~\ref{appx:pseudocode}.

To characterize the statistical validity of this decision rule, Appx.~\ref{appx:theory} formalizes the standard conformal assumption that, under $H_0$, the candidate poisoned pair and the reference samples in $\mathcal{D}_p$ are exchangeable. Under this assumption, the conformal $p$-value in Eq.~(12) follows a discrete uniform distribution, implying that \[ \Pr\!\left(p_{(x,t)} \leq \epsilon \mid H_0\right) \leq \epsilon. \] Therefore, the probability of misclassifying a poisoned pair as benign is controlled at the prescribed level $\epsilon$. The formal assumption, theorem, and proof are presented in Appx.~\ref{appx:theory}.

\begin{table*}[t]
\centering
\caption{FPR@100\%TPR (\%) and AUROC of \alg and competing backdoor defense methods on CC3M against 9 backdoor attacks. Avg. is the average result over 9 attacks. The best and runner-up results are \textbf{bolded} and \underline{underlined}.}
\label{tab:auroc}

\begingroup
\footnotesize
\setlength{\tabcolsep}{9.5pt}
\renewcommand{\arraystretch}{1.05}
\scalebox{0.8}{%
\begin{tabular}{lrcccccccccc}
\toprule
\textbf{Method}
&
& \textbf{BadNets}
& \textbf{Blended}
& \textbf{Trojan}
& \textbf{ISSBA}
& \textbf{LC}
& \textbf{WaNet}
& \textbf{mmPoison}
& \textbf{BadCLIP}
& \textbf{SIG}
& \textbf{\textsc{Avg.}} \\
\midrule

& & \multicolumn{10}{c}{FPR@100\%TPR ($\downarrow$)} \\
\cmidrule(lr){3-12}

ABL
& \textit{NeurIPS'21}
& 99.48
& 90.01
& 72.69
& 75.86
& 69.94
& 89.78
& 75.14
& 57.18
& 92.60
& 80.30 \\

CLIPScore
& \textit{ICML'23}
& 34.11
& 29.11
& 32.69
& 55.86
& 40.26
& 49.01
& 55.14
& 64.77
& 51.60
& 45.84 \\

SafeCLIP
& \textit{ICML'24}
& 52.62
& 74.06
& 68.39
& 62.33
& 67.50
& 46.54
& 70.14
& 66.53
& 49.23
& 61.93 \\

DAO
& \textit{ICLR'25}
& \textbf{2.06}
& \underline{3.33}
& \underline{5.49}
& \underline{8.14}
& \underline{18.94}
& \underline{7.32}
& \underline{11.13}
& \textbf{8.68}
& \underline{5.59}
& \underline{7.85} \\

\rowcolor{gray!15}
\alg (Ours)
&
& \underline{3.54}
& \textbf{3.20}
& \textbf{4.05}
& \textbf{3.71}
& \textbf{11.51}
& \textbf{3.63}
& \textbf{3.77}
& \underline{15.44}
& \textbf{3.22}
& \textbf{5.79} \\

\midrule
\midrule

& & \multicolumn{10}{c}{AUROC ($\uparrow$)} \\
\cmidrule(lr){3-12}

ABL
& \textit{NeurIPS'21}
& 0.5786
& 0.6007
& 0.5429
& 0.5439
& 0.5858
& 0.5607
& 0.5658
& 0.6826
& 0.5688
& 0.5811 \\

CLIPScore
& \textit{ICML'23}
& \underline{0.9964}
& \underline{0.9948}
& \underline{0.9967}
& \underline{0.9953}
& 0.7864
& \underline{0.9936}
& \underline{0.9915}
& 0.9010
& \underline{0.9965}
& \underline{0.9614} \\

SafeCLIP
& \textit{ICML'24}
& 0.8342
& 0.5701
& 0.7812
& 0.5512
& 0.5654
& 0.6031
& 0.8582
& 0.7186
& 0.8215
& 0.7004 \\

DAO
& \textit{ICLR'25}
& 0.9527
& 0.8450
& 0.9676
& 0.9422
& \underline{0.8274}
& 0.9559
& 0.9710
& \textbf{0.9424}
& 0.9606
& 0.9294 \\

\rowcolor{gray!15}
\alg (Ours)
&
& \textbf{0.9999}
& \textbf{0.9999}
& \textbf{0.9999}
& \textbf{0.9999}
& \textbf{0.9453}
& \textbf{0.9999}
& \textbf{0.9999}
& \underline{0.9358}
& \textbf{0.9999}
& \textbf{0.9867} \\

\bottomrule
\end{tabular}%
}
\endgroup
\end{table*}

\begin{table*}[t]
\centering
\caption{CA (\%) and ASR (\%) comparisons of backdoor defenses against 9 backdoor attacks on ImageNet-1K. Avg. is the average result over 9 attacks. \algi refers to \alg w/o Stage~\Rmnum{2}. The best and runner-up results are \textbf{bolded} and \underline{underlined}.}
\label{tab:DL_CA}

\begingroup
\footnotesize
\setlength{\tabcolsep}{9.7pt}
\renewcommand{\arraystretch}{1.05}
\scalebox{0.8}{%
\begin{tabular}{lrcccccccccc}
\toprule
\textbf{Method}
&
& \textbf{BadNets}
& \textbf{Blended}
& \textbf{Trojan}
& \textbf{ISSBA}
& \textbf{LC}
& \textbf{WaNet}
& \textbf{mmPoison}
& \textbf{BadCLIP}
& \textbf{SIG}
& \textbf{\textsc{Avg.}} \\
\midrule

& & \multicolumn{10}{c}{CA ($\uparrow$)} \\
\cmidrule(lr){3-12}

Clean Training
&
& 59.69 & 59.69 & 59.69 & 59.69 & 59.69
& 59.69 & 59.69 & 59.69 & 59.69 & 59.69 \\

No Defense
&
& 58.69 & 59.56 & 59.74 & 58.48 & 58.28
& 59.26 & 58.62 & 58.60 & 58.87 & 58.90 \\

\hdashline[0.5pt/1pt]

\rule{0pt}{10pt}CleanCLIP
& \textit{ICCV'23}
& 53.72 & 54.29 & 54.95 & 54.14 & 55.74
& 54.79 & 53.62 & 53.98 & 53.68 & 54.32 \\

RoCLIP
& \textit{NeurIPS'23}
& 40.37 & 44.81 & 43.78 & 44.00 & 42.09
& 49.03 & 47.47 & 49.18 & 45.26 & 45.11 \\

CleanerCLIP
& \textit{TIFS'25}
& 52.64 & 52.61 & 53.17 & 54.91 & 51.76
& 52.61 & 53.14 & 51.29 & 52.36 & 52.72 \\

DAO
& \textit{ICLR'25}
& \underline{57.89}
& \underline{56.56}
& \underline{56.36}
& \underline{55.28}
& \textbf{56.34}
& 56.26
& 56.25
& 54.61
& 55.11
& \underline{56.07} \\

InverTune
& \textit{NDSS'26}
& 57.52
& 54.49
& 55.57
& 53.98
& \underline{55.91}
& \underline{56.76}
& \underline{56.84}
& \underline{56.52}
& \underline{55.72}
& 55.92 \\

\rowcolor{gray!15}
\algi (Ours)
&
& 31.74 & 28.93 & 29.42 & 26.90 & 32.14
& 33.15 & 30.76 & 28.46 & 29.82 & 30.15 \\

\rowcolor{gray!15}
\alg (Ours)
&
& \textbf{59.21}
& \textbf{57.99}
& \textbf{58.99}
& \textbf{59.44}
& 55.11
& \textbf{59.67}
& \textbf{57.11}
& \textbf{57.41}
& \textbf{58.57}
& \textbf{58.17} \\

\midrule
\midrule

& & \multicolumn{10}{c}{ASR ($\downarrow$)} \\
\cmidrule(lr){3-12}

No Defense
&
& 100.0 & 100.0 & 93.11 & 50.28 & 83.58
& 99.35 & 17.79 & 98.85 & 80.38 & 80.37 \\

\hdashline[0.5pt/1pt]

\rule{0pt}{10pt}CleanCLIP
& \textit{ICCV'23}
& 17.13 & 18.43 & 21.16 & 4.13 & 0.01
& 5.49 & \textbf{0.00} & 89.60 & 21.72 & 19.74 \\

RoCLIP
& \textit{NeurIPS'23}
& 2.36 & 0.33 & 5.64 & 4.95 & \textbf{0.00}
& 0.67 & \textbf{0.00} & 47.20 & 4.23 & 7.26 \\

CleanerCLIP
& \textit{TIFS'25}
& 0.45 & 12.84 & 9.24 & 1.06 & 4.12
& 0.15 & 11.16 & 17.85 & 1.41 & 6.48 \\

DAO
& \textit{ICLR'25}
& \textbf{0.00} & 0.60 & 1.62 & 10.15 & \textbf{0.00}
& 0.50 & 8.73 & 11.62 & 0.20 & 3.71 \\

InverTune
& \textit{NDSS'26}
& 0.02
& 0.14
& 0.10
& 1.42
& 0.46
& 0.11
& 1.14
& \underline{0.49}
& 0.28
& \underline{0.46} \\

\rowcolor{gray!15}
\algi (Ours)
&
& \textbf{0.00}
& \textbf{0.00}
& \textbf{0.00}
& \textbf{0.00}
& 36.72
& \textbf{0.00}
& \textbf{0.00}
& 89.27
& \textbf{0.00}
& 14.00 \\

\rowcolor{gray!15}
\alg (Ours)
&
& \textbf{0.00}
& \textbf{0.00}
& \textbf{0.00}
& \textbf{0.00}
& 0.84
& \textbf{0.00}
& \textbf{0.00}
& \textbf{0.26}
& \textbf{0.00}
& \textbf{0.12} \\

\bottomrule
\end{tabular}%
}
\endgroup
\end{table*}

\section{Experiments}
\label{sec:experiments}
\subsection{Experimental Setup}
\label{subsec:setup}

\noindent\textbf{Network Architectures.} 
Following prior defenses, we use CLIP-RN50, comprising a
ResNet-50 image encoder and a Transformer text
encoder~\cite{CLIP}.


\noindent\textbf{Datasets.} 
We experiment on CC3M~\cite{sharma2018cc3m} for main comparison, and report the results on Flickr-PASCAL~\cite{pmlr-v202-yang23f}, Visual Genome~\cite{krishna2017visual}, and MSCOCO~\cite{mscoco}. 
To test the zero-shot classification performance on downstream datasets, we choose ImageNet-1K~\cite{imagenet}, Caltech-101~\cite{caltech101}, CIFAR-10/100~\cite{krizhevsky2009learning}, Oxford-IIIT PET~\cite{oxford_pets}, Stanford Cars~\cite{6755945}, and Food-101~\cite{food101}.


\noindent\textbf{Attack Settings.} 
Following previous works~\cite{bansal2023cleanclip}, we set the poison rate to $0.5\%$.
We evaluate nine widely used backdoor attacks, including: BadNets~\cite{gu2017badnets}, Blended~\cite{chen2017blended}, SIG~\cite{barni2019SIG}, LC~\cite{label-consist}, Trojan~\cite{Trojan}, WaNet~\cite{WaNet}, ISSBA~\cite{li2021SSBA}, mmPoison~\cite{pmlr-v202-yang23f} and BadCLIP~\cite{liang2023badclip}. 
Unless otherwise specified, we use \emph{banana} as the
default target class.

\noindent\textbf{Baseline Defenses.}
For overall evaluation, we compare \alg with multimodal backdoor detection and defense methods, including ABL~\cite{ABL}, CLIPScore~\cite{pmlr-v202-yang23f}, CleanCLIP~\cite{bansal2023cleanclip}, RoCLIP~\cite{RoCLIP}, SafeCLIP~\cite{SAFECLIP}, CleanerCLIP~\cite{CleanerCLIP}, DAO~\cite{DAO}, and InverTune~\cite{InverTune}.
We re-implement these defenses using official codes under same settings.
For comprehensive comparison, we also report results of \alg with only Stage~I, denoted as \algi.

\noindent\textbf{Evaluation Metrics.} 
We report TPR, FPR@100\%TPR, and AUROC for detection, and
zero-shot classification accuracy (CA) and ASR for downstream utility and security.
For further experimental setup, please refer to Appx.~\ref{appx:further-exp-setup}. 


\subsection{Detection Quality of \alg}
\label{exp:detection}

Tab.~\ref{tab:auroc} reports the detection performance of \alg in terms of FPR@100\%TPR and AUROC, clearly surpassing existing solutions. 
Our \alg consistently achieves accurate detection across all attacks.
Across various attacks, \alg lowers FPR from 29.11\%--64.77\% to 3.20\%--15.44\%, demonstrating its effective false positive control.
The results strongly demonstrate that \alg achieves the lowest average FPR@100\%TPR and obtains the best per-attack FPR on seven of the nine attacks. 
\alg achieves an AUROC of 0.9999 on seven attacks, while obtaining 0.9453 and 0.9358 on LC and BadCLIP, respectively.
We present the PR curves for \alg and \algi in Appx.~\ref{appdenx:precision_recall}. 
Results under standard and high poison rates are provided in Appx.~\ref{app:stree_test_poison_rate}. Across poison rates from $0.3\%$ to $0.6\%$, \alg consistently achieves TPRs above $98\%$ and outperforms all competing methods. Even when the poison rate reaches $10\%$, \alg maintains AUROC values between $0.9678$ and $0.9810$ across all evaluated attacks. 
Additional results further show that \alg generalizes across diverse and multiple target classes, different CLIP variants, and various pre-training datasets, as reported in Appx.~\ref{app:target_generalization},~\ref{appendix:clip_variant}, and~\ref{appendix:more_datasets}, respectively. 
Runtime comparison is provided in Appx.~\ref{ap:runtime}.

\subsection{Downstream Defense Performance}
\label{subsubsec:defense_result}

Tab.~\ref{tab:DL_CA} reports downstream performance, where ``Clean Training'' denotes training on benign data and ``No Defense'' denotes directly training on poisoned data. On ImageNet-1K, \alg achieves an average CA of $58.17\%$, close to clean training ($59.69\%$), while reducing the average ASR from $80.37\%$ to $0.12\%$. In contrast, \alg-I obtains only $30.15\%$ CA and $14.00\%$ ASR since Stage~I prioritizes high recall and a low FPR and retains only a limited number of benign pairs. 
Stage~II recovers benign samples from $\mathcal{D}_u$ while filtering residual poisoned pairs, substantially improving both utility and security. 
Additional evaluations in Appx.~\ref{appendix:downstream} on more datasets further confirm that the utility preservation of \alg generalizes beyond ImageNet-1K.

\subsection{Results of \alg against Adaptive Attacks}
\label{sec:adaptive_attacks}

We construct a defense-aware adaptive attack that explicitly targets the detection signals of both stages of \alg. To weaken Stage~I, we optimize a universal visual trigger such that triggered source images encode target semantics. Specifically, we construct a target text prototype by averaging the normalized text embeddings of target-class captions and optimize the trigger $\delta$ to minimize the cosine distance between each embedding $F_I(x_i+\delta)$ and this prototype. An $\ell_\infty$ regularization term is added to constrain the trigger magnitude. Overall optimization enhances the semantic consistency between triggered images and targeted poisoned captions, reducing the cross-modality discrepancy exploited by Stage~I.
To weaken Stage~II, we diversify the poisoned captions using two strategies. Under \textit{Paraphrase}, each original template-based poisoned caption is rewritten while preserving its target semantics. Under \textit{Description}, the original templates are replaced with independently generated natural-language descriptions of the target class, yielding greater linguistic diversity. Both transformations disperse poisoned captions in the text embedding space and weaken the compact target-specific structure exploited by the textual NCS. The captions are generated before trigger optimization and remain fixed throughout, with the $\delta$ as the only optimized variable.

Accordingly, we evaluate three adaptive attacks. \textit{Semantic Only} combines the optimized semantic trigger with the original poisoned captions. Building upon it, \textit{Semantic + Paraphrase} and \textit{Semantic + Description} replace the original captions with paraphrased captions and independently generated target-class descriptions, respectively. Detailed objectives and implementation settings are provided in Appx.~\ref{appx:adaptive-attacks}.



\begin{figure*}[t]
    \centering
    \begin{subfigure}[t]{0.31\textwidth}
        \centering
        \includegraphics[width=\linewidth]{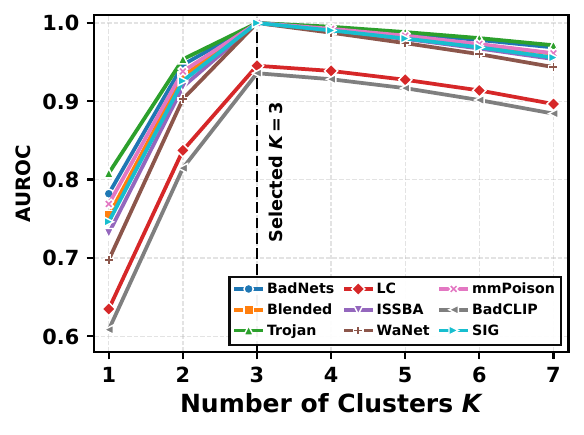}
        \caption{Impact of $K$.}
        \label{fig:sensitivity_k}
    \end{subfigure}
    \hfill
    \begin{subfigure}[t]{0.31\textwidth}
        \centering
        \includegraphics[width=\linewidth]{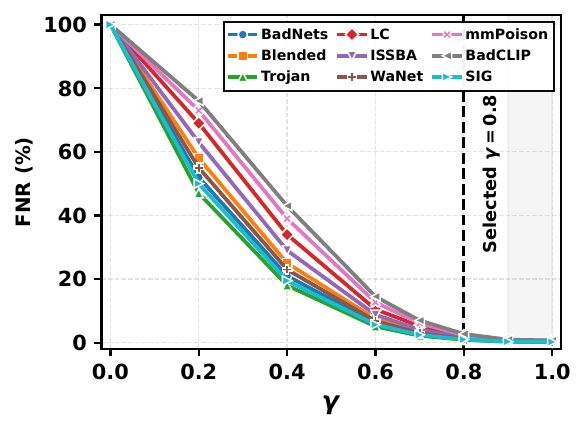}
        \caption{Impact of $\gamma$.}
        \label{fig:sensitivity_gamma}
    \end{subfigure}
    \hfill
    \begin{subfigure}[t]{0.31\textwidth}
        \centering
        \includegraphics[width=\linewidth]{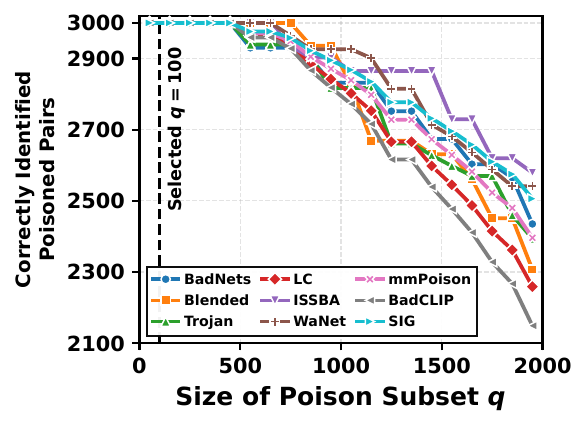}
        \caption{Impact of $q$.}
        \label{fig:sensitivity_q}
    \end{subfigure}
    \caption{Sensitivity of \alg to (a) the number of Gaussian
    components $K$, (b) the confidence threshold $\gamma$, and
    (c) the size of the poisoned reference subset $q$.}
    \label{fig:sensitivity}
\end{figure*}

\begin{table}[t]
\centering
\caption{Detection performance of \alg against the non-adaptive attack and three defense-aware adaptive attacks.}
\label{tab:adaptive-detection}

\begingroup
\setlength{\tabcolsep}{8pt}
\renewcommand{\arraystretch}{1.05}
\scalebox{0.8}{%
\begin{tabular}{@{}lcc@{}}
\toprule
\textbf{Attack Variant}
& \textbf{AUROC} ($\uparrow$)
& \shortstack{\textbf{FPR@100\%TPR}($\downarrow$)} \\
\midrule
Non-Adaptive            & 0.9867 & 5.79  \\
Semantic Only           & 0.9577 & 8.76  \\
Semantic + Paraphrase   & 0.9475 & 10.99 \\
Semantic + Description  & 0.9385 & 11.38 \\
\bottomrule
\end{tabular}%
}
\endgroup
\end{table}

As reported in Tab.~\ref{tab:adaptive-detection}, \alg remains highly effective against all three adaptive attacks. Compared with the non-adaptive attack, the AUROC of \alg decreases by only 0.029, 0.039, and 0.048, respectively. The corresponding FPR@100\%TPR increases by only 3.69--6.31\% and remains below 11.4\% in all cases. 
These results demonstrate that \alg remains highly effective even under strong defense-aware adaptations. 
We conjecture that the robustness of Stage~I stems from its reliance on a generation-induced representation space introduced by the external model $G$, rather than  on the semantic embedding space targeted by the adaptive attack. The imperfect alignment between these two spaces prevents the semantic-consistency attack from eliminating the discrepancy of poisoned pairs in the auxiliary space, thereby preserving their separability.

\subsection{Ablation Studies} 
\label{sec:ablation study}

\noindent\textbf{Component Number $K$ and Threshold $\gamma$.} 
Fig.~\ref{fig:sensitivity_k} reports the AUROC results with various $K$. 
When $K = 1$, the GMM degenerates into a single Gaussian, yielding relatively low AUROC between 0.6 and 0.8 and highlighting the necessity of modeling the distribution with multiple Gaussian components.
Increasing $K$ substantially improves the AUROC, which peaks consistently at $K = 3$.
However, further increasing $K$ over-partitions the distribution and gradually degrades detection performance.
As shown in Fig.~\ref{fig:sensitivity_gamma}, the proportion of mistakenly flagged poisoned samples decreases as $\gamma$ increases and becomes stable around $\gamma=0.8$. A larger $\gamma$ would leave almost no samples in $\mathcal{D}_b$. 
Therefore, we find that setting $K=3$ strikes an effective balance between separation and coverage, and we set $\gamma = 0.8$ to prioritize the purity of $\mathcal{D}_b$.

\noindent\textbf{Poisoned Subset Size $q$.} 
We examine whether increasing $q$ in the coarse-grained stage enhances detection performance across diverse attacks.
As shown in Fig.~\ref{fig:sensitivity_q}, \alg identifies all poisoned pairs for $q\in[50,400]$, whereas larger $q$ gradually degrades performance because benign samples increasingly contaminate $\mathcal{D}_p$ and distort the textual relevance estimate. We therefore set $q=100$.
Reference-set purity across different $q$ and robustness to deliberate contamination are reported in Appx.~\ref{appendix:purity} and ~\ref{appendix:fault_tolerance}, respectively.


\noindent\textbf{Caption Generation Model $G$.}
\label{ablation: diff-aug}
We examine the impact on CA and ASR with different $G$s, including ClipCap~\cite{clipcap_2021}, SmallCap~\cite{Ramos_2023_CVPR}, Oscar~\cite{OSCAR}, and CaMEL~\cite{barraco2022camel}. 
The ClipCap-TF denotes the transformer variant of ClipCap.
As shown in Tab.~\ref{ablation:clipcap}, \alg maintains high CA and zero ASR across all generators, despite differences in architecture and training data. 
This demonstrates its robustness to the choice of $G$ and the associated domain shift, a challenge also studied in \citet{wong2026knnproxy}.
Visualization of the visual-guided text embeddings is provided in Appx.~\ref{app:visualization}.


\begin{table}[t]
\centering
\caption{CA (\%) and ASR (\%) of \alg using different caption generation models $G$.}
\label{ablation:clipcap}

\begingroup
\setlength{\tabcolsep}{1.5pt}
\renewcommand{\arraystretch}{1.05}
\scalebox{0.72}{%
\begin{tabular}{@{}lcccccccc@{}}
\toprule
\multirow{2}{*}{\textbf{Model}}
& \multicolumn{2}{c}{\textbf{BadNets}}
& \multicolumn{2}{c}{\textbf{Blended}}
& \multicolumn{2}{c}{\textbf{Trojan}}
& \multicolumn{2}{c}{\textbf{WaNet}} \\
\cmidrule(lr){2-3}
\cmidrule(lr){4-5}
\cmidrule(lr){6-7}
\cmidrule(lr){8-9}
& CA ($\uparrow$) & ASR ($\downarrow$)
& CA ($\uparrow$) & ASR ($\downarrow$)
& CA ($\uparrow$) & ASR ($\downarrow$)
& CA ($\uparrow$) & ASR ($\downarrow$) \\
\midrule
CLIPCap-MLP
& 58.41 & 0.00
& 58.67 & 0.00
& 59.32 & 0.00
& 59.99 & 0.00 \\

CLIPCap-TF
& 59.30 & 0.00
& 59.43 & 0.00
& 58.94 & 0.00
& 59.38 & 0.00 \\

SmallCap
& 58.78 & 0.00
& 59.26 & 0.00
& 59.75 & 0.00
& 58.18 & 0.00 \\

Oscar
& 59.77 & 0.00
& 58.29 & 0.00
& 58.16 & 0.00
& 59.81 & 0.00 \\

CaMEL
& 59.94 & 0.00
& 59.46 & 0.00
& 58.35 & 0.00
& 59.33 & 0.00 \\
\bottomrule
\end{tabular}%
}
\endgroup
\end{table}

\section{Conclusion}
\label{sec:conclusion}

In this paper, we address the challenge of backdoor detection in MCL by introducing \alg, a novel coarse-to-fine two-stage backdoor detection method. 
Our approach overcomes the limitations of existing methods by employing a two-stage detection pipeline that dynamically adapts to the distribution of poisoned and benign samples while leveraging an effective NCS within the CP framework. 
Extensive experiments on large-scale datasets such as CC3M show that \alg significantly outperforms SOTA defenses, achieving high TPR across diverse attack scenarios.
Moreover, \alg remains robust against defense-aware adaptive attacks.
Despite these strengths, \alg focuses on detection, leaving post-hoc purification through model unlearning and concept erasure~\cite{li2026aegis,li2026llm} for future work.

\bibliography{aaai2027}



\twocolumn[{
\centering
\LARGE
\textbf{When Modalities Fail to Tango: Conformal Backdoor Detection \\ in Multimodal Contrastive Learning}\vspace{15pt}\\
\Large
\textbf{Supplementary Materials}\vspace{30pt} \\
}]

\appendix

\setcounter{page}{1}
\pagenumbering{arabic}


\makeatletter
\renewcommand{\addcontentsline}[3]{%
  \addtocontents{#1}{\protect\contentsline{#2}{#3}{\thepage}{}\protected@file@percent}%
}
\makeatother

\etocdepthtag.toc{mtappendix}
\etocsettagdepth{mtchapter}{none}
\etocsettagdepth{mtappendix}{subsection}

\tableofcontents


\vspace{20pt}

\section{Additional Related Work}
\label{appendix:related_work}

\subsection{Backdoor Attacks and Defenses in Multimodal Contrastive Learning}
\label{appendix:related_work_p1}
Backdoor attacks and defenses have been extensively studied across diverse scenarios~\cite{10285514,10494544,10552303,10620313,chen2024progressive,10491118,10304184}.
In MCL, attackers typically poison image--caption pairs by injecting visual triggers and replacing the associated captions with target semantics~\cite{liang2023badclip,carlini2022poisoning}.
Representative attacks include mmPoison~\cite{pmlr-v202-yang23f}, which studies multiple multimodal poisoning strategies; BadCLIP~\cite{liang2023badclip}, which optimizes visual triggers in the joint embedding space; and Dormant Backdoor~\cite{li2026dormant}, which implants latent model-level backdoors activated by downstream fine-tuning.
Existing defenses mitigate such threats through clean-data fine-tuning~\cite{bansal2023cleanclip}, surrogate-model-based data filtering~\cite{pmlr-v202-yang23f}, nearest-neighbor realignment~\cite{RoCLIP}, local-neighborhood sparsity~\cite{DAO}, or multimodal trigger inversion and activation tuning~\cite{InverTune}.
Although these methods mitigate backdoor threats to some extent, they remain insufficiently effective in practice, leaving substantial room for improvement.

\subsection{Conformal Anomaly and Outlier Detection}
\label{appendix:related_work_p2}

CP is a distribution-free, model-agnostic framework that quantifies the uncertainty of predictions from any ML model, offering robust guarantees under mild assumptions.
In recent years, CP has been widely applied to anomaly detection across various domains~\cite{laxhammar2011sequential,DBLP:journals/amai/LaxhammarF15,bashari2025robust}. 
Generally, CP-based anomaly detection methods operate by evaluating the NCS of each data point, which measures its deviation from a reference distribution. 
Data points with notably high NCS values are then flagged as potential anomalies.
For instance,~\citet{liang_integrative_2024} propose a CP-based approach for out-of-distribution (OOD) detection with labeled outliers; 
\citet{xiang2023cbd} develop a certified backdoored model detector through an adjustable CP scheme;
and~\citet{bashari2025robust} propose a data-cleaning framework to annotate data points in the contaminated reference sets in conformal OOD detection.

Formally, the general conformal anomaly detection paradigm~\cite{laxhammar2011sequential,DBLP:journals/amai/LaxhammarF15} is defined as follows. 
Given a $l$-sample reference set \(\left.\mathcal{D}_{\mathrm{ref}} = \{x_i\}_{i=1}^l\right.\) and a new point \(x_{l+1}\) to be evaluated, a carefully designed scoring function $A(\cdot)$ is employed to compute the NCS \(\left.\alpha_{l+1} = A(x_{l+1};\mathcal{D}_{\mathrm{ref}})\right.\). 
Intuitively, $\alpha_{l+1}$ quantifies how much $x_{l+1}$ deviates from established normality within $\mathcal{D}_{\mathrm{ref}}$.
Common scoring functions include those based on \(k\)-nearest neighbors, ridge regression, SVMs, random forests, and deep generative models such as VAE and SVDD~\cite{vovk2005algorithmic}. 
Once the NCS is computed, a conformal \textit{p}-value \(p_{l+1}\) for $x_{l+1}$ is derived. 
For a predefined significance threshold \(\epsilon \in (0,1)\), \(x_{l+1}\) is flagged as an anomaly if \(p_{l+1} \le \epsilon\). 
This mechanism inherently guarantees an FPR bounded by $\epsilon$, ensuring statistically rigorous and interpretable detection.
Building on these developments, we further extend CP’s capabilities to the MCL scenario, focusing on enhancing the quality of backdoor detection.

\section{Pseudocode of \alg}\label{appx:pseudocode}
Algorithm~\ref{alg:algorithm} presents the complete pseudocode of \alg.
\begin{algorithm}[ht] 
\small

\SetKwData{Left}{left}\SetKwData{This}{this}\SetKwData{Up}{up} \SetKwFunction{Union}{Union}\SetKwFunction{FindCompress}{FindCompress} \SetKwInOut{Input}{Input}\SetKwInOut{Output}{Output}
	
	\Input{Pre-trained image encoder $F_I$ and text encoder $F_T$; pre-trained mapping network $G$; training epoch $T$; training dataset $\tilde{\mathcal{D}}=\left\{(x_i, t_i)\right\}_{i=1}^N$; poisoned subset size $q$; probability threshold $\gamma$; detection threshold $\epsilon$; number of Gaussian components in GMM $K$.} 
    
	\Output{Identified benign subset $ \mathcal{D}_{b}$ and poisoned subset $ \mathcal{D}_{p}$; well-trained encoders $F_I$ and $F_T$.}
	 \BlankLine 
	    \tcc{Coarse-grained detection}
            \ForEach{$(x_i,t_i) \in \tilde{\mathcal{D}}$}{
            Produce generated text embedding $\hat{t}_{i}^e$\tcp*[r]{Eq.~\eqref{eq:geneate_caption}}
            Calculate cross-modality similarity $\kappa_i^{\rm orig},\kappa_i^{\rm gen}$;\\
            Calculate cross-modality consistency $\Delta^\kappa_i$\tcp*[r]{Eq.~\eqref{eq: cm consistency}}
            }
            
            Fit GMM with $K$ components on $\{\Delta^\kappa_i\}_{i=1}^{N}$\tcp*[r]{Eq.~\eqref{eq:gmm}}
            $i^\star \gets$ Gaussian component with minimum mean\tcp*[r]{Eq.~\eqref{eq:min-component}}
            $\mathcal{D}_b \gets$ pairs w.p. higher than $\gamma$ in $i^\star$\tcp*[r]{Eqs.~\eqref{eq:responsibility}–\eqref{eq:db}}
            $\mathcal{D}_p \gets$ pairs with top-$q$ $\Delta^\kappa$ values\tcp*[r]{Eq.~\eqref{eq:dp}}
            $\mathcal{D}_u \gets \tilde{\mathcal{D}} \setminus (\mathcal{D}_p \cup \mathcal{D}_b)$;\\
            
            \BlankLine
            \tcc{Fine-grained detection}
            Construct reference set $\mathcal{T}_{\rm ref}$\tcp*[r]{Eq.~\eqref{eq:reference-set}}
            Compute NCS $\{\alpha_i\}_{i=1}^q$ in $\mathcal{T}_{\rm ref}$\tcp*[r]{Eq.~\eqref{eq:ncs}}
            \ForEach{$(x_i,t_i) \in \mathcal{D}_u$}{
            Calculate \textit{p}-value $p_{(x_i,t_i)}$\tcp*[r]{Eq.~\eqref{eq:p_value}}
            \uIf {$p_{(x_i,t_i)} \le \epsilon$} {
                $\mathcal{D}_b \gets \mathcal{D}_b \cup \{(x_i, t_i)\}$;
            } \Else {
                $\mathcal{D}_p \gets \mathcal{D}_p \cup \{(x_i, t_i)\}$;
            }
            }
            \BlankLine
            \tcc{Train CLIP model}
            \For{$t = 1$ \KwTo $T$}{
                Train $F_I, F_T$ on $\mathcal{D}_b$;
            }
 	 	  \caption{\textbf{\alg}}
 	 	  \label{alg:algorithm} 
\end{algorithm}

\section{Theoretical Analysis}
\label{appx:theory}

We now establish the theoretical guarantee showing that our CP-based fine-grained detection provably controls the type-I error, i.e., the probability of misclassifying a poisoned sample as benign.
The statistical validity of this guarantee relies on the exchangeability assumption between the test pair and the reference set.
In our context, this means that under the null hypothesis $H_0$ that the candidate is poisoned, the test pair $(x,t)$ and the reference pairs in $\mathcal{D}_p$ can be regarded as i.i.d. samples from the same poisoned distribution.
Under this condition, the ordering of all NCSs is invariant under any permutation, which directly leads to uniformity of the conformal \textit{p}-value.
We formalize this as follows.

\begin{assumption}[Exchangeability]\label{ass1}
Assume that $(x,t)$ and the samples in $\mathcal{D}_p$ are drawn i.i.d. from the poisoned distribution. 
Then, the joint distribution of the NCS values is invariant under any permutation $\pi$ over $\{1,\dots,q+1\}$, that is,
\[
(\alpha_1,\dots,\alpha_{q+1}) \overset{d}{=} (\alpha_{\pi(1)},\dots,\alpha_{\pi(q+1)}),
\]
where $\overset{d}{=}$ denotes equality in distribution.
\end{assumption}

\begin{remark}
Assumption~1 is a reasonable approximation when the selected reference set $\mathcal{D}_p$ has high purity, since its samples are then predominantly drawn from the poisoned distribution.
In practice, Tab.~\ref{tab:percentage_poison_in_top_q} in Appx.~\ref{appendix:purity} shows that $\mathcal{D}_p$ consistently exhibits high empirical purity, measured as the proportion of truly poisoned samples in $\mathcal{D}_p$, across various backdoor attacks and choices of $q$. Consequently, although the top-$q$ selection procedure does not strictly produce i.i.d.\ samples from the poisoned distribution, the resulting reference set can be regarded as an approximate realization of this idealized setting due to its limited contamination by benign samples. 
Thus, Assumption~1 should be viewed as an idealized condition for the subsequent analysis, while its practical plausibility and the robustness of \alg to moderate violations are supported empirically.


\end{remark}

We now state the main theoretical result as follows,
\begin{theorem}[Type-I Error Control]
\label{thm}
Under Assumption~\ref{ass1}, if the test pair $(x,t)$ is drawn from the poisoned distribution (i.e., $H_0$ holds), then the conformal \textit{p}-value defined in Eq.~(\ref{eq:p_value}) is uniformly distributed over $\big\{\frac{1}{q+1},\frac{2}{q+1},\dots,1\big\}$. Consequently, for any significance level $\epsilon\in (0,1)$, 
\[
\Pr\big(p_{(x,t)} \le \epsilon \mid H_0\big) \le \epsilon.
\]
\end{theorem}

\begin{proof}
Under Assumption~\ref{ass1}, all $q+1$ samples used for conformal calibration, including the $q$ reference samples and the test pair $(x,t)$, are conditionally exchangeable.
Thus, the joint distribution of their NCS values $\{\alpha_1,\dots,\alpha_{q+1}\}$ is invariant under any permutation of indices.
This property guarantees that the relative position (rank) of the test score $\alpha_{q+1}$ among the $q+1$ values is uniformly random, regardless of its actual numerical value.
Let $r$ denote this rank when sorted in decreasing order; exchangeability ensures that each of the $q+1$ possible ranks is equally probable, i.e.,
\begin{equation}
\Pr\,(r = k) = \frac{1}{q+1}, \quad \text{for}\ k = 1,2,\dots,q+1.
\end{equation}

By the definition of the conformal \textit{p}-value in Eq.~\eqref{eq:p_value}, the empirical tail probability associated with the test pair is
\begin{equation}
p_{(x,t)} = \frac{r}{q+1}.
\end{equation}

Therefore, $p_{(x,t)}$ follows a discrete uniform distribution over $\big\{\frac{1}{q+1}, \frac{2}{q+1}, \dots, 1\big\}$.
Now, for any $\epsilon \in (0,1)$, 
\begin{equation}
\Pr\,(p_{(x,t)} \le \epsilon) = \Pr\,(r \le \epsilon (q+1))= \frac{\lfloor \epsilon (q+1) \rfloor}{q+1} \le \epsilon.
\end{equation}
\end{proof}

\begin{remark}
Thm.~\ref{thm} has a clear implication: if we set $\epsilon=0.05$, then even in the worst case, at most $5\%$ of poisoned pairs can be incorrectly accepted as benign.
This guarantee holds \emph{without making strong distributional assumptions}---the only requirement is approximate exchangeability, which our coarse-grained stage empirically satisfies.
Hence, the fine-grained CP stage provides a statistically principled safeguard: it ensures that false negatives (i.e., poisoned samples slipping through as benign) are strictly controlled by a user-chosen error tolerance $\epsilon$.
\end{remark}

\section{Partition Pipeline}
\label{appx:partition}
Fig.~\ref{fig:GMM} illustrates the partitioning pipeline.

\begin{figure*}[t]
\centering
\includegraphics[width=\linewidth]{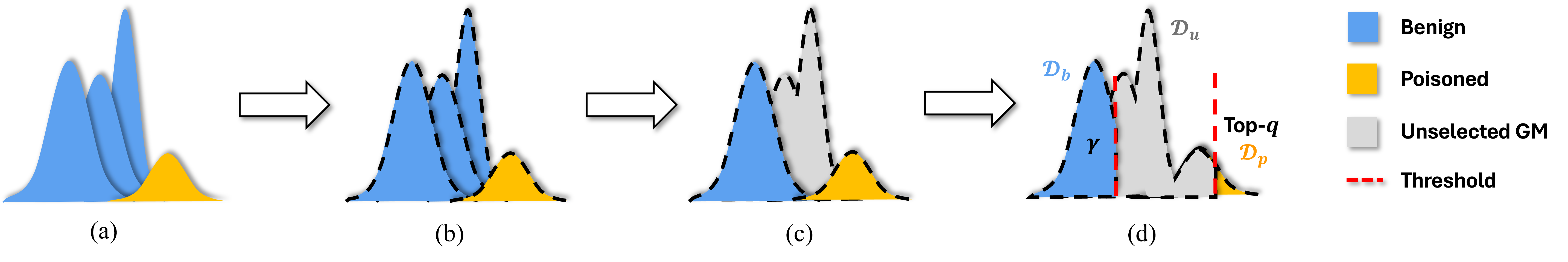}
\caption{
\textbf{Illustration of GMM-based subset partitioning process.} 
(a) Distribution of consistency values. 
(b) A GMM is fitted to the consistency distribution, with each Gaussian component outlined by a dotted line. 
(c) The components with the minimum and maximum means ($\mu$) are highlighted. 
(d) Samples from the leftmost Gaussian component are selected with probability $\gamma$ to form the benign subset $\mathcal{D}_b$ (Eq.~\eqref{eq:db}), while the top-$q$ pairs are grouped into the poisoned subset $\mathcal{D}_p$ (Eq.~\eqref{eq:dp}). 
The remaining samples constitute the unidentified subset $\mathcal{D}_u$.
}
\label{fig:GMM}
\end{figure*}


\section{Further Experimental Setup}\label{appx:further-exp-setup}

\subsection{Datasets}
\label{appendix:dataset}
In line with prior methods~\cite{bansal2023cleanclip, RoCLIP, SAFECLIP, pmlr-v202-yang23f, InverTune}, for pre-training, we experiment on CC3M~\cite{sharma2018cc3m} for main comparison, and also report \alg's results on Flickr-PASCAL~\cite{pmlr-v202-yang23f}, Visual Genome~\cite{krishna2017visual}, and MSCOCO~\cite{mscoco}.
CC3M is a dataset consisting of $\sim$3.3M image-caption pairs scraped from the web.
Flickr-PASCAL is merged by Flickr30k (abbr. Flickr)~\cite{young-etal-2014-image} and PASCAL~\cite{10.5555/1866696.1866717} following Yang~\citet{pmlr-v202-yang23f}. 
For the downstream datasets used to test the zero-shot classification performance of the trained CLIP models, we use ImageNet-1K~\cite{imagenet} for main comparison, and also report \alg's results on Caltech-101~\cite{caltech101}, CIFAR-10/100~\cite{krizhevsky2009learning}, Oxford-IIIT PET~\cite{oxford_pets}, Stanford Cars~\cite{6755945}, and Food-101~\cite{food101}.

\subsection{Evaluation Metrics}
\label{appendix:metrics}
To evaluate detection quality, we report the TPR, FPR@100\%TPR and AUROC.
TPR and FPR are defined by \( \textrm{TPR} = \frac{|\mathcal{D}_p \cap \tilde{\mathcal{D}}'|}{|\tilde{\mathcal{D}}'|} \), the proportion of poisoned samples correctly flagged, and \( \textrm{FPR} = \frac{|\mathcal{D}_p \cap (\tilde{\mathcal{D}}\setminus \tilde{\mathcal{D}}')|}{|\tilde{\mathcal{D}}\setminus \tilde{\mathcal{D}}'|} \), the proportion of benign samples mistakenly flagged. 
FPR at high TPR reports the FPR at a single high TPR (e.g., 100\% TPR), which allows for a quick review of the false alarm of detection when all poisoned pairs are detected.
For evaluating the detection as a defense, we report zero-shot classification accuracy (CA, a high retained CA demonstrates that the detection method effectively preserves benign samples while filtering out backdoors) on the ImageNet-1K validation set. 
To assess defense effectiveness, we evaluate the ASR (a significant ASR reduction indicates successful prevention of backdoors through poisoned sample isolation), which quantifies the proportion of backdoored images misclassified as the target class by a model. 
When computing ASR, we exclude the target class while injecting triggers into images from other classes. 

\subsection{Implementation Setup}
\label{sec:other details}
In \alg, the model is trained solely on the identified benign subset $\mathcal{D}_b$ for 10 epochs with a batch size of 128, using the AdamW optimizer with a learning rate of 1e-6. 
We use a weight decay for all the parameters during training, except for batch/layer norm, bias, and logit scale parameters. 
All experiments are run on 4 NVIDIA 4090 GPUs with PyTorch.
We exclude the comparison with \cite{Token_Level_Unlearning, UBT_arxiv_2024} due to the unavailability of their source code and misalignment in experimental setups.

\section{Additional Results}
\label{appx:add-result}

\subsection{Precision--Recall Curve of \alg}
\label{appdenx:precision_recall}
We present PR curves for \alg and \algi in Fig.~\ref{fig:precision-recall}. 
As can be seen clearly from the results, \algi struggles to fully identify poisoned samples, achieving high precision only when recall is near 0.
In contrast, \alg exhibits high precision and recall across all attacks, with the AUPRC approaching 1.


\begin{figure}[t]
    \centering

    \begin{subfigure}[t]{0.48\linewidth}
        \centering
        \includegraphics[width=\linewidth]
        {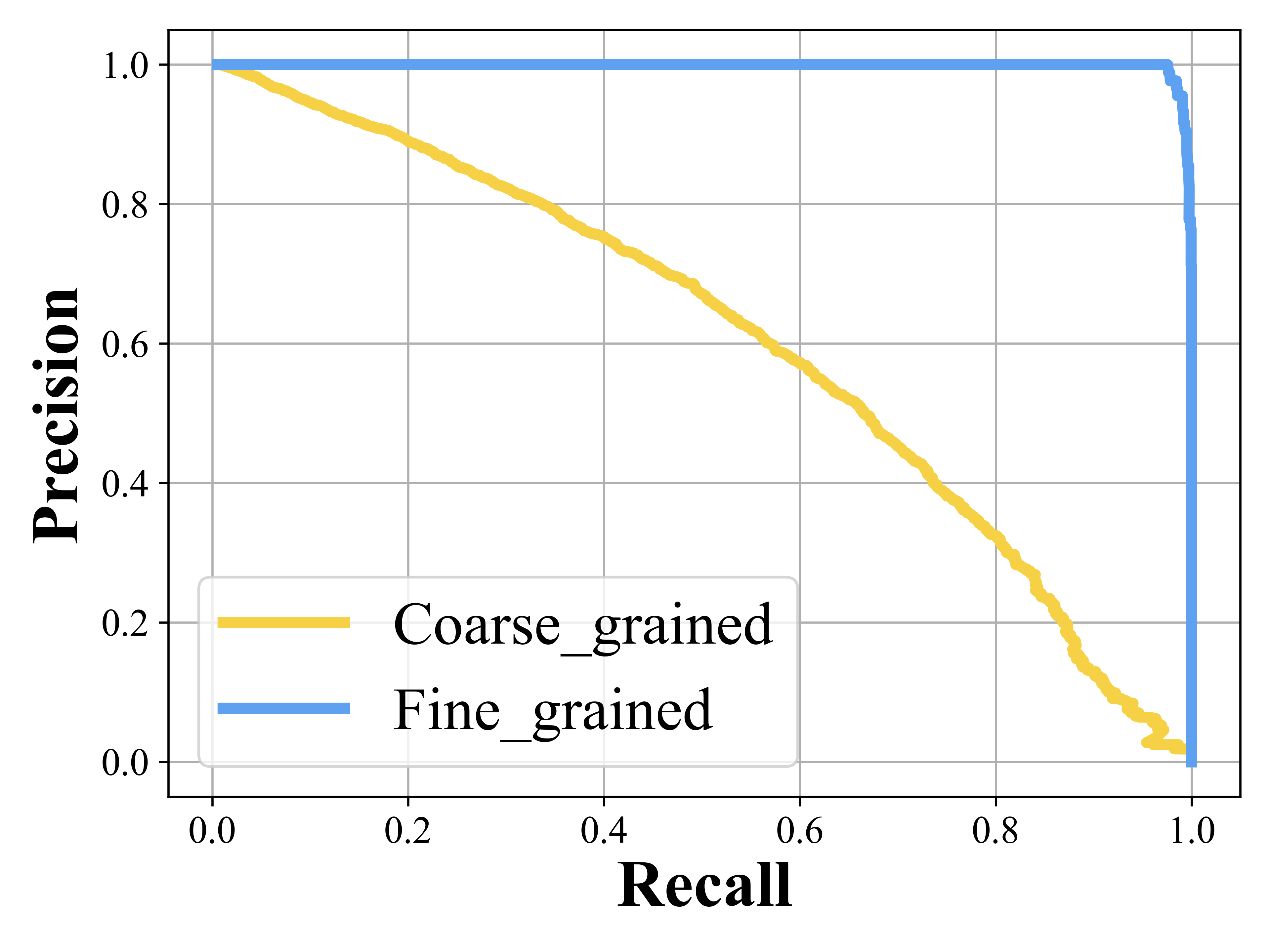}
        \caption{BadNets}
        \label{fig:pr_badnets}
    \end{subfigure}
    \hfill
    \begin{subfigure}[t]{0.48\linewidth}
        \centering
        \includegraphics[width=\linewidth]
        {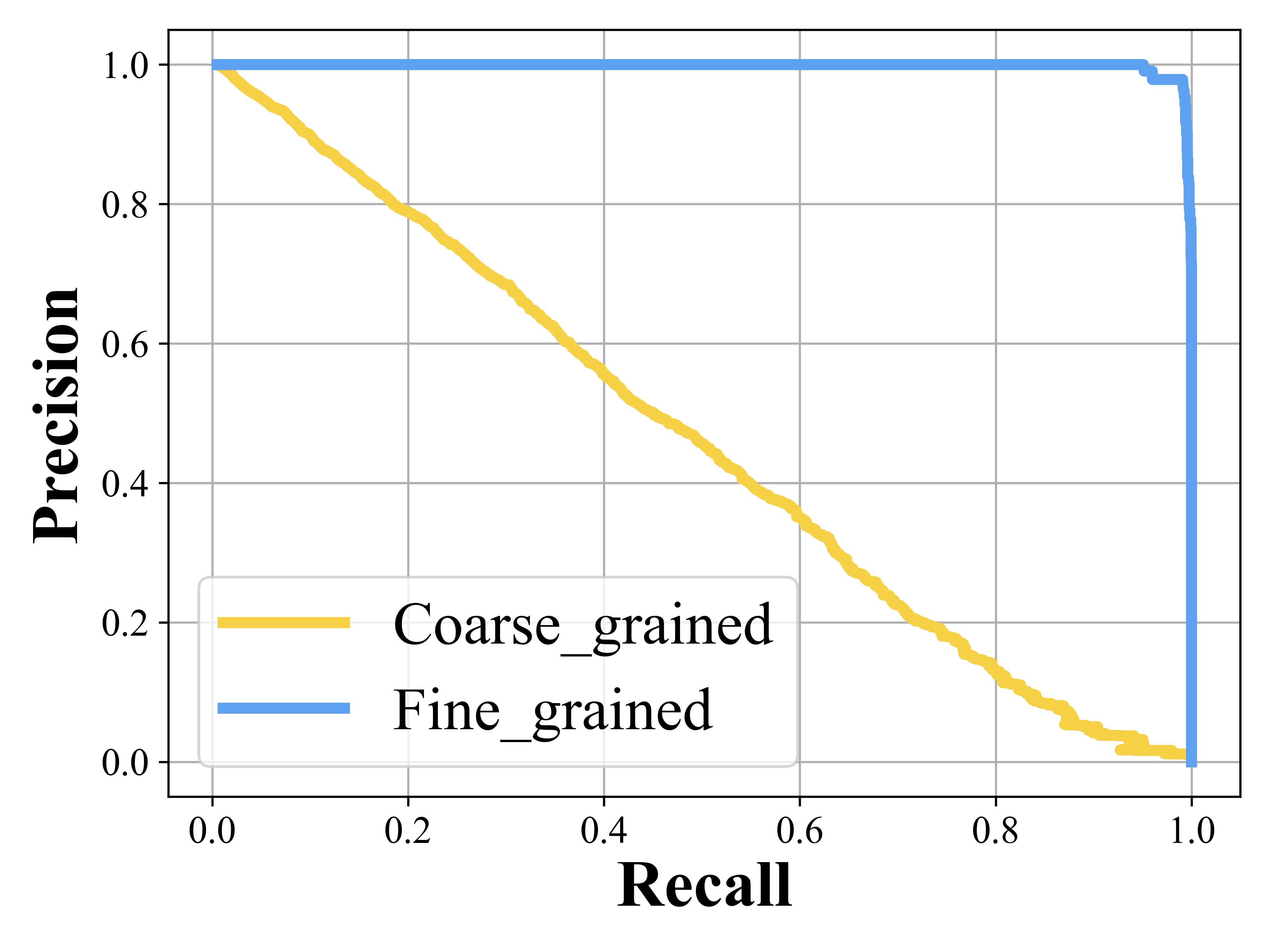}
        \caption{Blended}
        \label{fig:pr_blended}
    \end{subfigure}

    \vspace{2pt}

    \begin{subfigure}[t]{0.48\linewidth}
        \centering
        \includegraphics[width=\linewidth]
        {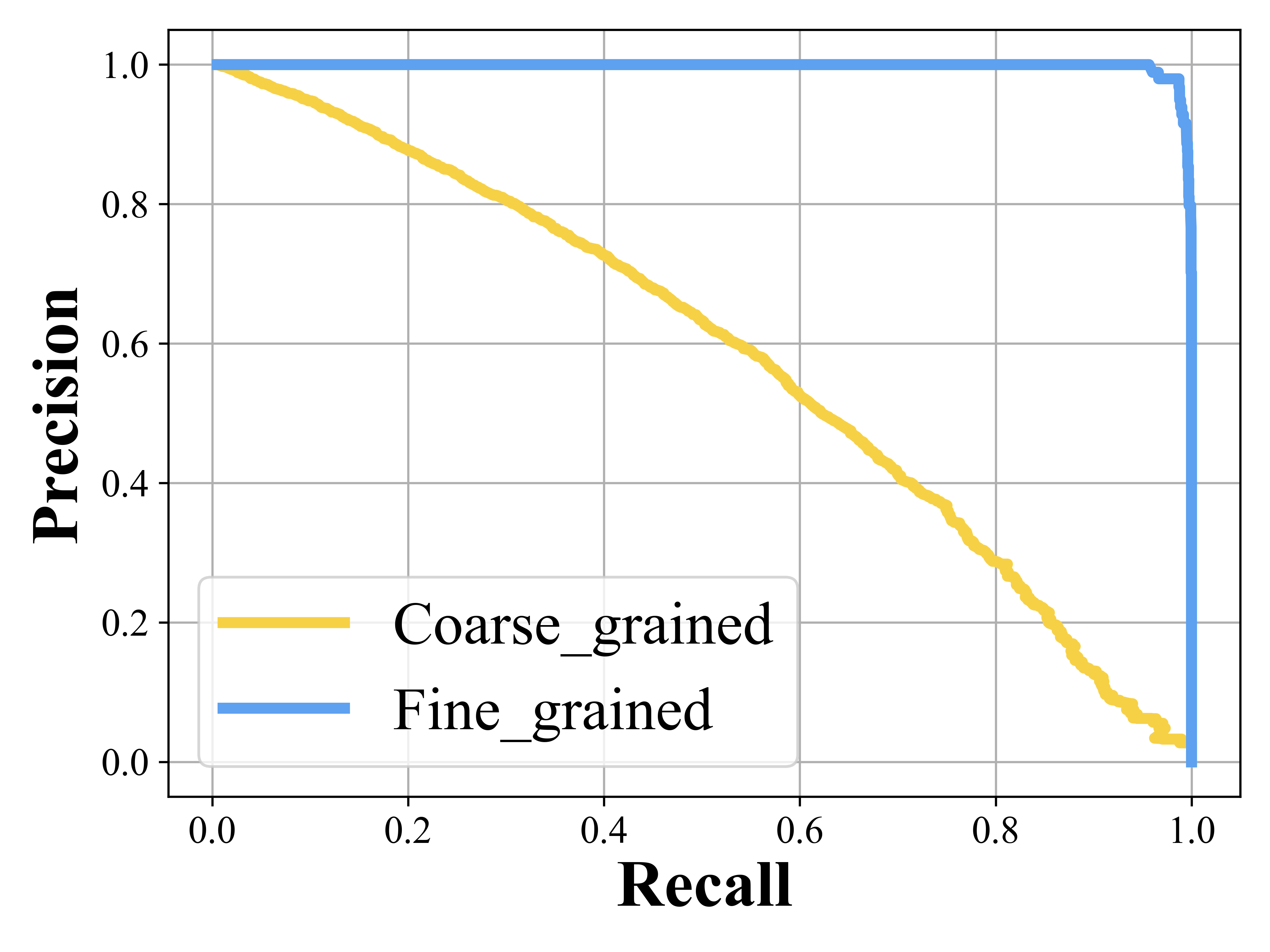}
        \caption{Trojan}
        \label{fig:pr_trojan}
    \end{subfigure}
    \hfill
    \begin{subfigure}[t]{0.48\linewidth}
        \centering
        \includegraphics[width=\linewidth]
        {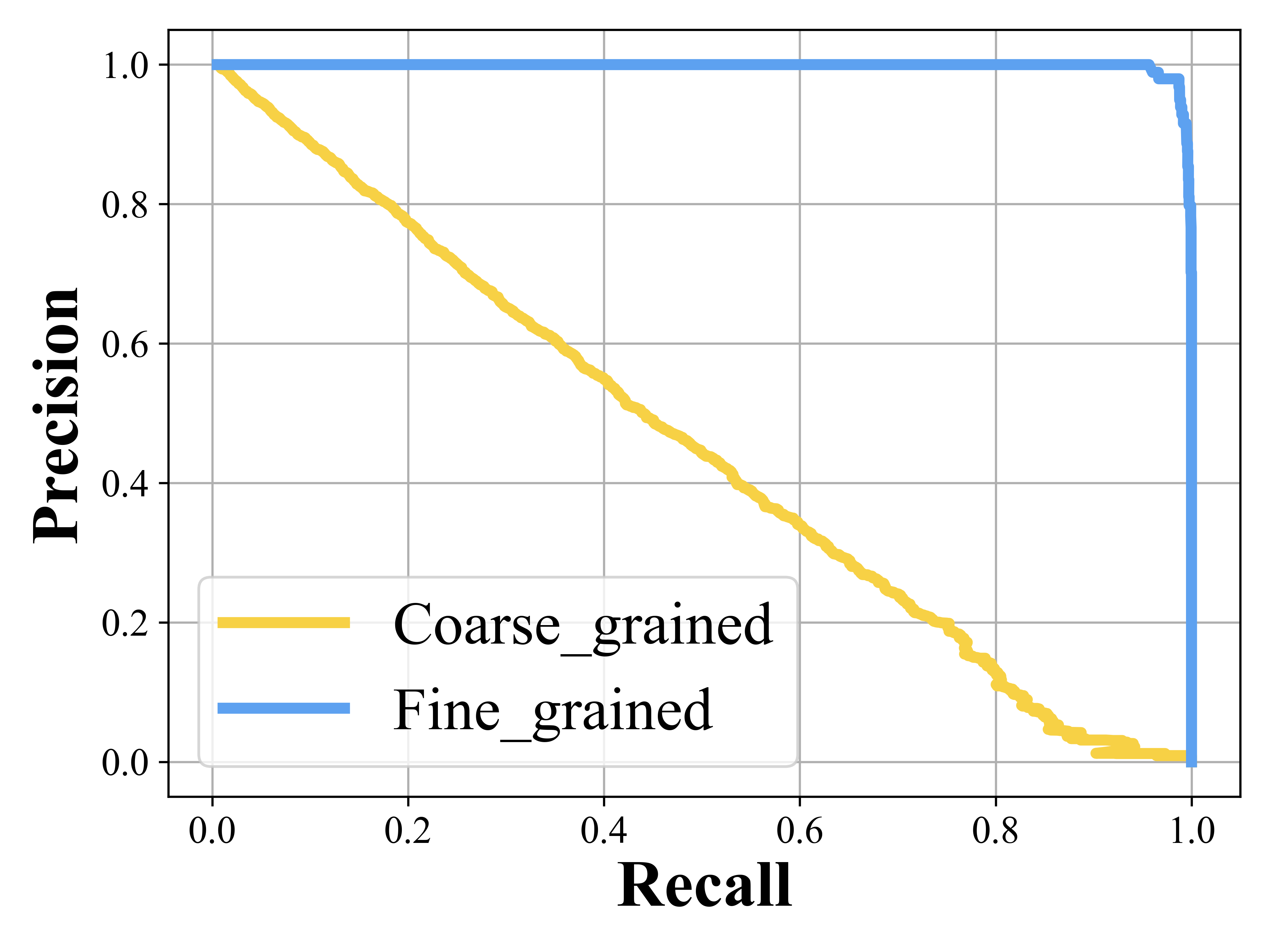}
        \caption{WaNet}
        \label{fig:pr_wanet}
    \end{subfigure}

    \vspace{2pt}

    \begin{subfigure}[t]{0.48\linewidth}
        \centering
        \includegraphics[width=\linewidth]
        {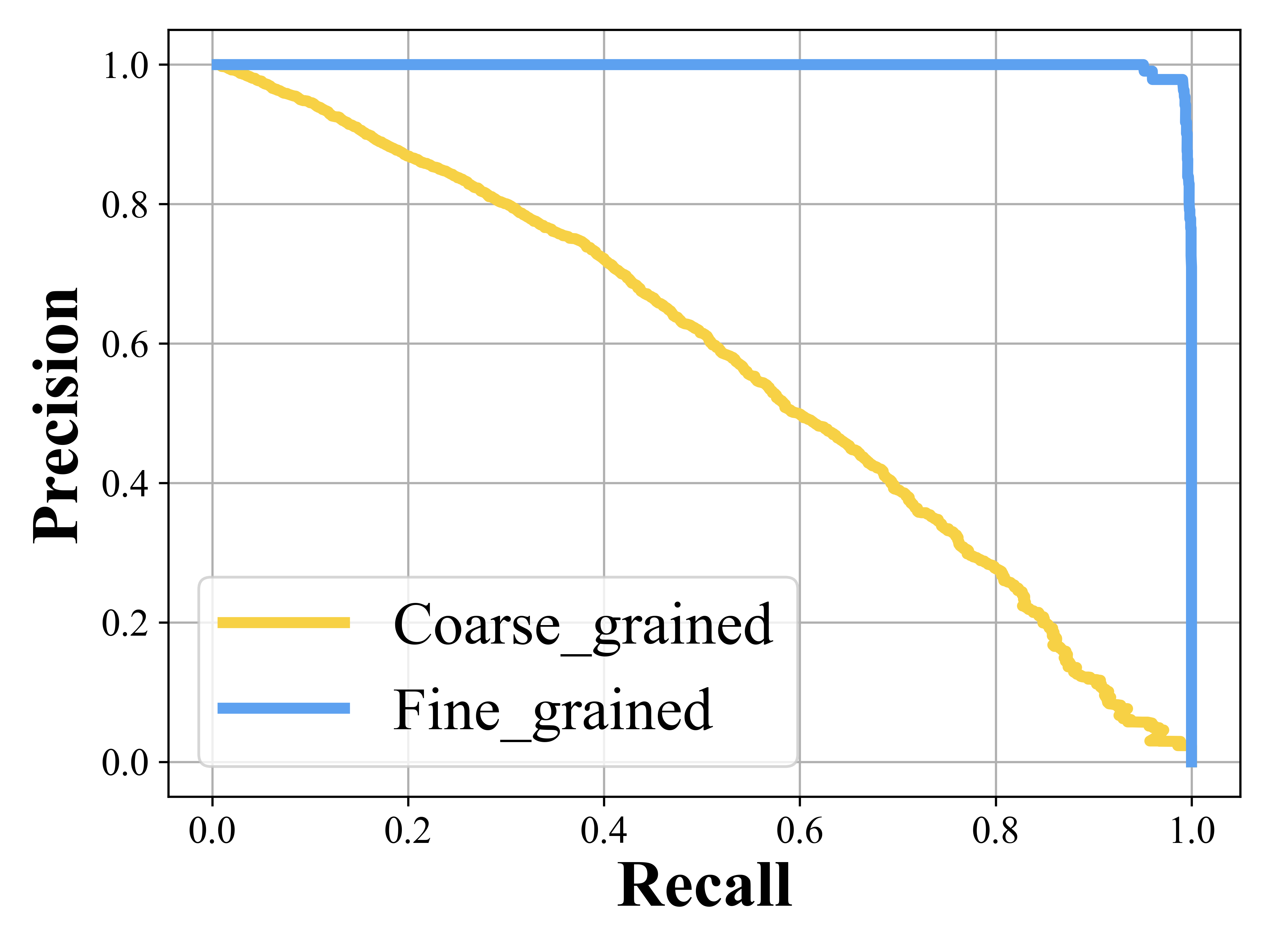}
        \caption{ISSBA}
        \label{fig:pr_issba}
    \end{subfigure}
    \hfill
    \begin{subfigure}[t]{0.48\linewidth}
        \centering
        \includegraphics[width=\linewidth]
        {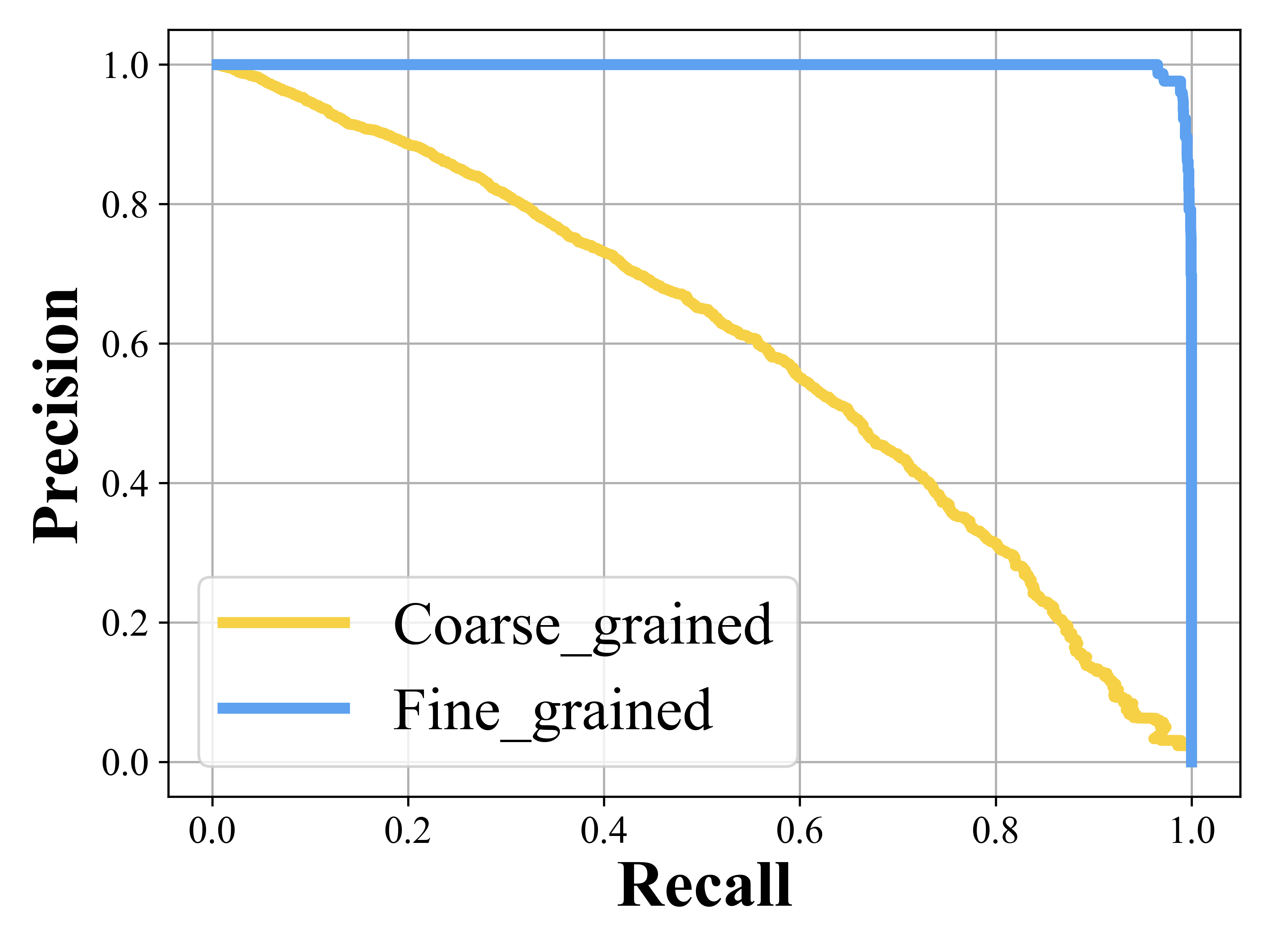}
        \caption{SIG}
        \label{fig:pr_sig}
    \end{subfigure}

    \caption{Precision--recall curves of the coarse-grained
    (\algi) and coarse-to-fine (\alg) detection methods across six
    attacks. \alg achieves an AUPRC close to 1 and consistently
    outperforms \algi.}
    \label{fig:precision-recall}
\end{figure}

\subsection{Stress Test of \alg at High Poison Rate}
\label{app:stree_test_poison_rate}

\begin{table}[t]
\centering
\caption{TPR (\%) of \alg and other methods on CC3M against various attacks under different poison rates $\rho$. The best and runner-up results are \textbf{bolded} and \underline{underlined}.}
\scalebox{0.74}{%
\begin{tabular}{@{}ll*{5}{c}}
\toprule
$\rho$ & \textbf{Method} & \textbf{BadNets} & \textbf{Blended} & \textbf{Trojan} & \textbf{ISSBA} & \textbf{WaNet} \\ \midrule

\multirow{5}{*}{0.6\%} 
& ABL & 54.66 & 49.26 & 43.87 & 17.16 & 24.84 \\
& CLIPScore & 92.45 & 86.63 & 90.67 & 87.21 & 89.79 \\
& SafeCLIP & 83.42 & 57.01 & 73.87 & 67.16 & 85.82 \\
& DAO & \underline{98.16} & \underline{96.36} & \underline{95.16} & \underline{93.16} & \underline{96.27} \\
& \cellcolor{gray!15}\alg (Ours) 
  & \cellcolor{gray!15}\textbf{98.46} 
  & \cellcolor{gray!15}\textbf{98.67} 
  & \cellcolor{gray!15}\textbf{99.69} 
  & \cellcolor{gray!15}\textbf{98.38} 
  & \cellcolor{gray!15}\textbf{98.28} \\ \midrule

\multirow{5}{*}{0.5\%} 
& ABL & 45.41 & 47.92 & 31.56 & 11.28 & 19.92 \\
& CLIPScore & 93.56 & 88.58 & 91.35 & 88.14 & 90.73 \\
& SafeCLIP & 81.45 & 62.72 & 77.15 & 63.86 & 88.02 \\
& DAO & \underline{98.16} & \underline{96.47} & \underline{96.15} & \underline{94.07} & \underline{96.44} \\
& \cellcolor{gray!15}\alg (Ours) 
  & \cellcolor{gray!15}\textbf{100.0} 
  & \cellcolor{gray!15}\textbf{98.96} 
  & \cellcolor{gray!15}\textbf{99.08} 
  & \cellcolor{gray!15}\textbf{98.23} 
  & \cellcolor{gray!15}\textbf{98.27} \\ \midrule

\multirow{5}{*}{0.4\%} 
& ABL & 73.25 & 41.77 & 26.91 & 8.27 & 17.28 \\
& CLIPScore & 93.14 & 88.62 & 90.68 & 89.44 & 90.07 \\
& SafeCLIP & 79.53 & 59.61 & 71.56 & 61.25 & 84.16 \\
& DAO & \underline{98.16} & \underline{96.52} & \underline{96.01} & \underline{93.96} & \underline{96.73} \\
& \cellcolor{gray!15}\alg (Ours) 
  & \cellcolor{gray!15}\textbf{98.49} 
  & \cellcolor{gray!15}\textbf{99.08} 
  & \cellcolor{gray!15}\textbf{99.02} 
  & \cellcolor{gray!15}\textbf{98.75} 
  & \cellcolor{gray!15}\textbf{98.02} \\ \midrule

\multirow{5}{*}{0.3\%} 
& ABL & 29.57 & 37.16 & 25.19 & 7.67 & 15.24 \\
& CLIPScore & 93.88 & 89.36 & 91.03 & 89.66 & 90.18 \\
& SafeCLIP & 72.14 & 54.61 & 74.55 & 65.08 & 81.67 \\
& DAO & \underline{98.16} & \underline{96.64} & \underline{95.72} & \underline{93.59} & \underline{96.66} \\
& \cellcolor{gray!15}\alg (Ours) 
  & \cellcolor{gray!15}\textbf{99.45} 
  & \cellcolor{gray!15}\textbf{98.23} 
  & \cellcolor{gray!15}\textbf{98.30} 
  & \cellcolor{gray!15}\textbf{98.63} 
  & \cellcolor{gray!15}\textbf{99.01} \\ 
\bottomrule
\end{tabular}%
}
\label{tab:detection_3m}
\end{table}

We further assess \alg at different poison rates (0.3\%--0.6\%), with results summarized in Tab.~\ref{tab:detection_3m}.
Overall, \alg consistently outperforms all other methods across all attack types and poison rates, achieving the highest TPR values in every category. Specifically, at a 0.6\% poison rate, \alg yields the highest TPRs, surpassing the second-best method DAO, especially in 
ISSBA attack. 
Even at lower poison rates (0.4\% and 0.3\%), \alg continues to excel, with TPRs consistently exceeding 98\% across all attacks.

To evaluate the resilience of \alg under varying levels of data poisoning, we report AUROC scores across diverse backdoor attacks and different poison rates in Tab.~\ref{stress_poison_rate}.

\begin{table}[t]
\centering
\caption{AUROC of \alg at different poison rates $\rho$.}
\scalebox{0.61}{%
\begin{tabular}{@{}lcccccccc}
\toprule
$\rho$  & \textbf{BadNets} & \textbf{Blended}& \textbf{Trojan}& \textbf{ISSBA}& \textbf{WaNet} & \textbf{SIG} & \textbf{mmPoison} & \textbf{BadCLIP} \\
\midrule
2\% & .9999 & .9999 & .9999 & .9999 & .9999 & .9999 & .9999 & .9932 \\
4\% & .9948 & .9972 & .9826 & .9951 & .9912 & .9956 & .9948 & .9906\\
6\% & .9870 & .9823 & .9759 & .9845 & .9856 & .9871 & .9870 & .9868\\
8\% & .9861 & .9712 & .9710 & .9746 & .9786 & .9751 & .9861 & .9845\\
10\% & .9752 & .9678 & .9697 & .9715 & .9687 & .9721 & .9752 & .9810\\
\bottomrule
\end{tabular}%
}
\label{stress_poison_rate}
\end{table}

Overall, \alg remains robust across a wide range of poison rates. At a poison rate of $2\%$, it achieves AUROC values between $0.9932$ and $0.9999$. When the poison rate increases to $4\%$, the AUROC ranges from $0.9826$ to $0.9972$. Although detection performance gradually decreases as the poison rate further increases, \alg still maintains AUROC values between $0.9678$ and $0.9810$ at a severe poison rate of $10\%$.

The detection results of \alg against low poison rate to high poison rate highlights the effectiveness of \alg in robustly detecting backdoor attacks regardless of the poison rate and attack type.


\subsection{Generalization to Diverse and Multiple Target Classes}
\label{app:target_generalization}

The main experiments use ``banana'' as the default target class. We first replace it with ``volcano'' to evaluate whether the detection performance depends on the semantics of a specific target. We then construct two multiple-target settings with target sets $\mathcal{Y}_2=\{\text{banana},\text{airplane}\}$ and $\mathcal{Y}_3=\{\text{banana},\text{airplane},\text{volcano}\}$. For each conventional attack, poisoned pairs associated with different targets are simultaneously injected into the same pre-training dataset. 

Beyond extending conventional single-target attacks to multiple targets, we further evaluate \alg against MTAttack~\cite{wang2026mtattack}, a dedicated multi-target backdoor attack that jointly optimizes multiple triggers and establishes distinct trigger--target correspondences. Specifically, MTAttack employs Proxy Space Partitioning and Trigger Prototype Anchoring constraints to separate different triggers in the latent visual space and reduce inter-trigger interference. We adapt its generated poisoned image--caption pairs to our MCL pre-training setting and evaluate \alg.





As shown in Tab.~\ref{tab:target_generalization}, \alg is insensitive to the use of a particular target class. When the default target is replaced by ``volcano'', it achieves AUROC values between $0.993$ and $0.999$ across all six attacks. This demonstrates that the detector does not rely on target-specific keywords or the semantics of ``banana''.

\begin{table}[t]
\centering
\caption{AUROC of \alg under diverse single-target and multi-target settings on CC3M. Here, $\mathcal{Y}_1=\{\textit{volcano}\}$, $\mathcal{Y}_2=\{\textit{banana},\textit{airplane}\}$, and $\mathcal{Y}_3=\{\textit{banana},\textit{airplane},\textit{volcano}\}$.}
\label{tab:target_generalization}
\scriptsize
\setlength{\tabcolsep}{1.5pt}
\resizebox{\columnwidth}{!}{%
\begin{tabular}{lcccccccc}
\toprule
\textbf{Target Set}
& \textbf{BadNets}
& \textbf{Blended}
& \textbf{Trojan}
& \textbf{WaNet}
& \textbf{mmPoison}
& \textbf{BadCLIP}
& \textbf{MTAttack}
& \textbf{Avg.} \\
\midrule
$\mathcal{Y}_1$
& .999 & .999 & .999 & .999 & .999 & .993 & -- & .998 \\
$\mathcal{Y}_2$
& .999 & .999 & .999 & .999 & .999 & .992 & .985 & .996 \\
$\mathcal{Y}_3$
& .994 & .995 & .982 & .991 & .991 & .990 & .987 & .990 \\
\bottomrule
\end{tabular}%
}
\end{table}

\alg also generalizes well to multiple-target attacks. Under the two-target setting, the AUROC remains between $0.985$ and $0.999$. When the number of targets increases to three, the AUROC ranges from $0.982$ to $0.994$. Nevertheless, all attacks remain accurately detectable without modifying the detector or using knowledge of the target classes.

We further evaluate \alg against the dedicated MTAttack, which represents a stronger multi-target threat than directly combining independently constructed single-target attacks. \alg achieves an AUROC of 0.987 under the $\mathcal{Y}_3$-target setting. Although MTAttack explicitly enhances the separation of different trigger--target mappings, its poisoned image--caption pairs still exhibit the cross-modal and target-specific distributional structures exploited by \alg. Consequently, \alg remains effective without attack-specific adaptation, further demonstrating its generalizability to purpose-built multi-target attacks.


\subsection{Generalizability across CLIP Variants}
\label{appendix:clip_variant}
To verify the generalizability of \alg across architectures, we evaluate five CLIP-based models against seven attacks: BadNets, Blended, Trojan, ISSBA, WaNet, mmPoison, and BadCLIP.
As shown in Tab.~\ref{tab:variants}, all evaluated models achieve an ASR of $0.00\%$ across the seven attacks, demonstrating that \alg generalizes consistently across different image encoder architectures and pre-training paradigms.
In terms of clean accuracy, CLIP-ViT-B/16 achieves the highest average CA of $63.38\%$, followed by DeCLIP-ViT-B/32 with $63.02\%$. CLIP-ViT-B/32, DeCLIP-ResNet50, and SLIP-ViT-B/16 achieve average CA values of $59.47\%$, $60.25\%$, and $52.26\%$, respectively.
Overall, \alg proves effective and robust across different CLIP variants, ensuring consistent performance.

\begin{table*}[t]
\centering
\caption{CA (\%) and ASR (\%) of \alg with other CLIP variants on ImageNet-1K.}
\scalebox{0.66}{%
\begin{tabular}{l*{14}{c}}
\toprule
& \multicolumn{2}{c}{\textbf{BadNets}} & \multicolumn{2}{c}{\textbf{Blended}}  & \multicolumn{2}{c}{\textbf{Trojan}} & \multicolumn{2}{c}{\textbf{ISSBA}} & \multicolumn{2}{c}{\textbf{WaNet}} & \multicolumn{2}{c}{\textbf{mmPoison}} & \multicolumn{2}{c}{\textbf{BadCLIP}}\\
\cmidrule(lr){2-3} \cmidrule(lr){4-5} \cmidrule(lr){6-7} \cmidrule(lr){8-9} \cmidrule(lr){10-11} \cmidrule(lr){12-13} \cmidrule(lr){14-15}
\textbf{Model}          & CA ($\uparrow$)   & ASR ($\downarrow$)  & CA ($\uparrow$)   & ASR ($\downarrow$)  & CA ($\uparrow$)        & ASR ($\downarrow$) &  CA ($\uparrow$)        & ASR ($\downarrow$) &  CA ($\uparrow$)        & ASR ($\downarrow$)   &  CA ($\uparrow$)        & ASR ($\downarrow$) &  CA ($\uparrow$)        & ASR ($\downarrow$)   \\
\midrule
CLIP-ViT-B/32~\cite{CLIP} & 59.74 & 0.00 & 59.99 & 0.00 &  59.44 & 0.00 & 59.67 & 0.00 & 59.11 & 0.00  & 59.14 & 0.00 & 59.20 & 0.00\\
CLIP-ViT-B/16~\cite{CLIP} & 63.74 & 0.00 & 63.99 & 0.00 &  63.44 & 0.00 & 63.67 & 0.00 & 63.11 & 0.00 & 62.76 & 0.00 & 62.98 & 0.00\\
DeCLIP-ResNet50~\cite{li2022supervision}  & 60.50 & 0.00 & 60.12 & 0.00 & 60.33 & 0.00 & 60.27 & 0.00 & 60.77 & 0.00 & 59.86 & 0.00 & 59.89 & 0.00\\
DeCLIP-ViT-B/32~\cite{li2022supervision}  & 63.42 & 0.00 & 62.82 & 0.00 & 63.16 & 0.00 & 62.54 & 0.00 & 63.04 & 0.00 & 63.01 & 0.00 & 63.14 & 0.00\\
SLIP-ViT-B/16~\cite{slip}  & 52.16 & 0.00 & 51.37 & 0.00 & 53.10 & 0.00 & 53.44 & 0.00 & 53.25 & 0.00 & 51.21 & 0.00  & 51.26 & 0.00\\
\bottomrule
\end{tabular}%
}
\label{tab:variants}
\end{table*}

\subsection{Detection Results on More Pretrained Datasets}
\label{appendix:more_datasets}
Tab.~\ref{appendix:other_datasets} reports \alg's AUROC across three pre-training datasets for eight backdoor attacks, with all values ranging from $0.9554$ to $0.9921$, demonstrating consistent generalization across additional pre-training datasets.
Flickr-PASCAL performs marginally worse than MSCOCO and Visual Genome across most attacks. For SIG, its 0.9628 trails MSCOCO's 0.9875 and Visual Genome's 0.9844; similar gaps appear in Trojan and WaNet.
The disparity stems from Flickr-PASCAL's limited scale, task singularity, and suboptimal image-text alignment, which reduce cross-modality consistency separability.

\begin{table}[t]
\centering
\caption{AUROC of \alg on additional pre-training datasets.}
\label{appendix:other_datasets}
\small
\setlength{\tabcolsep}{5pt}
\renewcommand{\arraystretch}{1.05}
\begin{tabular}{@{}lccc@{}}
\toprule
\textbf{Attack}
& \textbf{Flickr-PASCAL}
& \textbf{MSCOCO}
& \textbf{Visual Genome} \\
\midrule
BadNets  & .9819 & .9846 & .9829 \\
Blended  & .9812 & .9916 & .9874 \\
Trojan   & .9748 & .9870 & .9871 \\
ISSBA    & .9759 & .9827 & .9921 \\
WaNet    & .9728 & .9847 & .9848 \\
SIG      & .9628 & .9875 & .9844 \\
mmPoison & .9554 & .9752 & .9765 \\
BadCLIP  & .9569 & .9726 & .9788 \\
\bottomrule
\end{tabular}
\end{table}

\subsection{Time Consumption}
\label{ap:runtime}

We report the runtime of \alg and existing defenses in Fig.~\ref{fig:time}. 
Compared with CLIP, \alg introduces a negligible overhead, requiring only 12 extra minutes for detection. 
RoCLIP introduces higher cost by retrieving similar samples for each image, extending training to nearly 14 hours.
SafeCLIP is the most time-consuming, exceeding 41 hours due to extensive data augmentation on $\sim$75\% of the training data.
Overall, \alg achieves efficiency through early-stage detection and by avoiding redundant computations, providing a lightweight yet effective defense paradigm.

\begin{figure}[t]
\centering
\includegraphics[width=\linewidth]{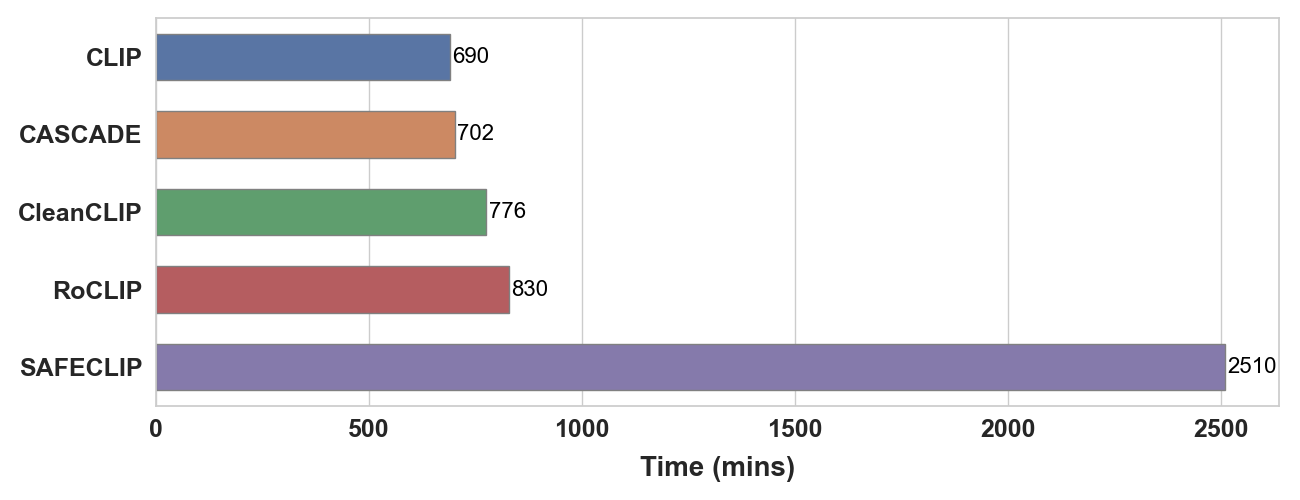}
\caption{Time consumption comparison of \alg and existing defense methods.}
\label{fig:time}
\end{figure}

\subsection{Downstream Performance Comparison on Other Datasets}
\label{appendix:downstream}
Tab.~\ref{exp:more_zero_shot} presents the zero-shot classification accuracy of various CLIP-based models across multiple downstream datasets. Our proposed method, \alg, consistently outperforms other robust CLIP variants (CleanCLIP and RoCLIP) while maintaining competitive performance relative to the original CLIP. On average, \alg achieves 66.1\% accuracy, surpassing CleanCLIP (+7.8\%) and RoCLIP (+6.0\%), though it trails the standard CLIP (72.0\%) due to its trade-offs for robustness. Notably, \alg excels on fine-grained datasets like Oxford-IIIT PET (80.1\%) and Caltech-101 (82.1\%), demonstrating its balanced capability. The results highlight \alg's superior generalizability under zero-shot transfer scenarios, suggesting effective mitigation of backdoor injection while preserving CLIP's foundational strengths.

\begin{table}[t]
\centering
\caption{Zero-shot classification performance (\%) on downstream datasets. 
Avg. is the average result over 6 datasets.
The best and runner-up results are \textbf{bolded} and \underline{underlined}.}
\scalebox{0.825}{%
\begin{tabular}{lccccccr}
\toprule
\textbf{Method} & 
\rotatebox{90}{\footnotesize \textbf{Food-101}} & 
\rotatebox{90}{\footnotesize \textbf{CIFAR-10}} & 
\rotatebox{90}{\footnotesize \textbf{CIFAR-100}} & 
\rotatebox{90}{\footnotesize \textbf{Oxford-IIIT PET}} & 
\rotatebox{90}{\footnotesize \textbf{Stanford Cars}} & 
\rotatebox{90}{\footnotesize \textbf{Caltech-101}} & 
\textsc{\textbf{Avg.}} \\
\midrule
CLIP     & 81.0 & 75.6 & 41.6 & 86.2 & 62.3 & 85.1 & 72.0 \\
\hdashline[0.5pt/1pt]
\rule{0pt}{10pt}CleanCLIP & \underline{65.0} & 51.4 & 23.1 & \underline{77.6} & \underline{52.1} & \underline{80.3} & 58.3 \\
RoCLIP    & 64.1 & \underline{59.0} & \underline{35.8} & 72.3 & 50.3 & 79.3 & \underline{60.1} \\
\rowcolor{gray!15}\alg (Ours)     & \textbf{75.6} & \textbf{64.8} & \textbf{38.0} & \textbf{80.1} & \textbf{55.8} & \textbf{82.1} & \textbf{66.1} \\
\bottomrule
\end{tabular}%
}
\label{exp:more_zero_shot}
\end{table}

\begin{table*}[t]
\centering
\caption{Downstream performance under the proposed defense-aware adaptive attacks, evaluated using zero-shot classification accuracy (CA) on ImageNet-1K. All values are percentages.}
\label{tab:adaptive-downstream}
\begin{tabular}{lcccc}
\toprule
Attack Variant
& CA w/o \alg ($\uparrow$)
& CA w/ \alg ($\uparrow$)
& ASR w/o \alg ($\uparrow$)
& ASR w/ \alg ($\downarrow$)\\
\midrule
Semantic Only
& 58.40 & 58.12 & 90.26 & 8.49\\
Semantic + Paraphrase
& 58.22 & 57.87 & 82.43 & 9.60\\
Semantic + Description
& 58.11 & 57.45 & 78.69 & 11.03\\
\bottomrule
\end{tabular}
\end{table*}

\subsection{Result of \alg against Adaptive Attacks}
\label{appx:adaptive-attacks}

\paragraph{Semantic Consistent Attack Design.}
We construct a stronger adaptive attack that combines semantic-consistency-aware visual trigger optimization with two fixed caption transformations. Importantly, the captions are generated before trigger optimization and remain frozen throughout the attack. The only optimized variable is the universal visual trigger. 
We consider three adaptive variants: semantic-consistency-aware trigger optimization with the original poisoned captions, paraphrased captions, and independently generated target-class descriptions.


We first propose two strategies to diversify the poisoned captions, namely \emph{paraphrasing} and \emph{target-description}. 
First, under the \emph{paraphrasing} strategy, we paraphrase each original poisoned caption while preserving its target semantics. This introduces diverse lexical choices and syntactic structures among captions associated with the same target. Second, under the \emph{target-description} strategy, we replace template-based poisoned captions with independently generated natural-language descriptions of the target class. Compared with fixed target-caption templates, these descriptions exhibit substantially greater semantic and linguistic diversity. Both strategies are designed to disperse poisoned captions in the text embedding space, thereby weakening the compact poisoned-caption distribution exploited by Stage II.

To construct a visual trigger that encodes the target-class semantics, we optimize the trigger such that the visual representation of a triggered source image is aligned with the textual representations of the poisoned captions. Let $F_I(\cdot;\theta_v)$ and $F_T(\cdot;\theta_t)$ denote the visual and text encoders, respectively, whose parameters are kept frozen during trigger optimization. Given a set of prepared poisoned captions $\mathcal{T}_t=\{t_j^t\}_{j=1}^{N_t}$, we first construct a target text prototype by averaging their normalized text embeddings:
\begin{equation}
\mu_t
=
\frac{1}{N_t}
\sum_{j=1}^{N_t}
\frac{F_T(t_j^t;\theta_t)}
{|F_T(t_j^t;\theta_t)|_2},
\qquad
\bar{\mu}_t
=
\frac{\mu_t}{|\mu_t|*2}.
\end{equation}
Given a source image $x_i$ and a learnable visual trigger $\delta$, we denote the corresponding triggered image by $\widetilde{x}_i$. The trigger is then optimized by minimizing the cosine distance between the visual representation of each triggered source image and the target text prototype:
\begin{equation}
\label{eq:visual_alignment}
\mathcal{L}_{\mathrm{vis}}
=
\frac{1}{N_s}
\sum_{i=1}^{N_s}
\left[
1-
\left(
\frac{F_I(\widetilde{x}_i;\theta_v)}
{|F_I(\widetilde{x}_i;\theta_v)|_2}
\right)^{\top}
\bar{\mu}_t
\right].
\end{equation}


To limit the perceptibility of the optimized trigger, we introduce a perturbation regularization term based on the commonly adopted $\ell_{\infty}$ perturbation budget $\epsilon=8/255$~\cite{DBLP:conf/nips/LiLW0024,11045541}: \begin{equation} \label{eq:perturbation_regularization}  \mathcal{L}_{\mathrm{pert}} = \left[ \max \left( 0, \|\delta\|_{\infty}-\epsilon \right) \right]^2, \qquad \epsilon=\frac{8}{255}. \end{equation}
The final trigger optimization objective is therefore formulated as \begin{equation} \label{eq:trigger_total_loss} \mathcal{L}_{\mathrm{total}} = \mathcal{L}_{\mathrm{vis}} + \lambda_{\mathrm{pert}} \mathcal{L}_{\mathrm{pert}}, \end{equation} where $\lambda_{\mathrm{pert}}$ controls the trade-off between target text-prototype alignment and trigger imperceptibility.
By aligning the visual representations of triggered images with the target text prototype and pairing them with captions semantically related to the same target class, the attack causes both modalities to consistently encode the target concept, thereby forming a semantic-consistency attack.

\paragraph{Result Analysis.} 
We evaluate \alg's robustness against adaptive attacks, as shown in Tab.~\ref{tab:adaptive-detection} and Tab.~\ref{tab:adaptive-downstream}. 
As shown in Tab.~\ref{tab:adaptive-detection}, the defense-aware adaptive attacks reduce the AUROC of \alg from 0.9867 to 0.9577--0.9385, indicating that explicitly optimizing image--caption semantic consistency makes poisoned pairs more difficult to detect. Nevertheless, \alg maintains strong detection performance across all three adaptive attacks. Among them, \textit{Semantic + Description} is more challenging than \textit{Semantic + Paraphrase}, resulting in a lower AUROC of 0.9385 and a higher FPR@100\%TPR of 11.38\%. Compared with paraphrases, independently generated target-class descriptions often contain generic attributes shared by multiple classes, such as fur being common to both cat and dog. Such common attributes weaken the class-specific structure of poisoned captions and increase their ambiguity. Despite this challenge, the textual nonconformity signal exploited in Stage~II remains effective, allowing \alg to reliably distinguish poisoned samples from benign ones.

Tab.~\ref{tab:adaptive-downstream} further shows that \alg substantially suppresses the downstream attack effectiveness while largely preserving clean accuracy. Specifically, the ASR after applying \alg remains between 8.49\% and 11.03\%, compared with 78.69\%--90.26\% without defense. Notably, \textit{Semantic + Description} achieves a lower undefended ASR than \textit{Semantic + Paraphrase}, i.e., 78.69\% versus 82.43\%. This is because natural target-class descriptions provide a weaker and less consistent poisoning signal than template-based target-class captions, making it more difficult for the model to learn a stable backdoor shortcut. Meanwhile, this weakened attack signal also makes the corresponding poisoned samples less distinguishable, which explains the moderately higher defended ASR of 11.03\%. Even under this most challenging variant, \alg removes the majority of the attack effect while retaining a CA of 57.45\%.


\begin{table}[t]
\centering
\caption{True poisoned rate (\%) in $\mathcal{D}_p$ with different $q$.}
\scalebox{0.775}{%
\begin{tabular}{@{}lcccccccc}
\toprule
& \multicolumn{8}{c}{$q$} \\
\cmidrule(lr){2-9}
\textbf{Attack}  & \textbf{50} & \textbf{100}& \textbf{150}& \textbf{200}& \textbf{250} & \textbf{300} & \textbf{350} & \textbf{400} \\
\midrule
BadNets & 100.0 & 96.6 & 96.0 & 95.5 & 93.6 & 92.0 & 91.4 & 90.7\\
Blended & 98.0 & 95.0 & 90.0 & 87.5 & 87.2 & 84.6 & 83.4 & 82.5\\
Trojan & 100.0 & 98.0 & 96.6 & 94.5 & 93.6 & 91.6 & 91.1 & 91.0\\
ISSBA & 100.0 & 96.6 & 96.0 & 93.0 & 92.8 & 91.6 & 90.8 & 89.5\\
WaNet & 92.0 & 92.0 & 87.3 & 85.5 & 84.8 & 83.0 & 82.2 & 80.7 \\
SIG & 100.0 & 97.0 & 97.3 & 94.5 & 93.6 & 92.6 & 92.0 & 91.5 \\
mmPoison & 100.0 & 97.2 & 96.3 & 94.9 & 93.7 & 92.3 & 92.0 & 91.4 \\
BadCLIP & 100.0 & 96.5 & 95.9 & 94.4 & 93.9 & 92.6 & 92.0 & 91.0 \\
\bottomrule
\end{tabular}%
}
\label{tab:percentage_poison_in_top_q}
\end{table}

\begin{table}[t]
\centering
\caption{AUROC of \alg when poisoned subset contains benign pairs, which evaluates the error tolerance of the fine-grained stage w.r.t. the quality of the poisoned subset.}
\scalebox{0.73}{%
\begin{tabular}{@{}lcccccccc}
\toprule
 & \multicolumn{4}{c}{\textbf{$q$=50}} & \multicolumn{4}{c}{\textbf{$q$=400}}\\
\cmidrule(lr){2-5}  \cmidrule(lr){6-9}
\textbf{Fault\% =}          & \textbf{20\%}   & \textbf{40\%}  & \textbf{60\%}   & \textbf{80\%}  & \textbf{20\%}        & \textbf{40\%} &  \textbf{60\%}        & \textbf{80\%} \\\midrule
BadNets & .9999 & .9995 & .9893 & .9217 & .9999 & .9995 & .9993 & .9905 \\
Blended & .9999 & .9993 & .9887 & .9174 & .9999 & .9994 & .9992 & .9911 \\
Trojan & .9999 & .9996 & .9885 & .9196  & .9999 & .9995 & .9989 & .9905 \\
ISSBA & .9999 & .9993 & .9901 & .9136 & .9999 & .9994 & .9993 & .9908 \\
WaNet & .9999 & .9987 & .9894 & .9184 & .9999 & .9994 & .9991 & .9909\\
SIG & .9999 & .9999 & .9887 & .9166 & .9999 & .9994 & .9992  & .9906 \\
mmPoison & .9999 & .9981 & .9845 & .9069 & .9999 & .9884 & .9802  & .9842 \\
BadCLIP & .9358 & .9127 & .8927 & .8752 & .9358 & .9248 & .9271  & .9191 \\
\bottomrule
\end{tabular}%
}
\label{tab:fault_tolerance}
\end{table}

\subsection{Purity of Reference Set}
\label{appendix:purity}

To validate the assumptions underlying CP, we assess the purity of the reference set \(\mathcal{D}_p\) obtained in the coarse-grained stage. 
As shown in Tab.~\ref{tab:percentage_poison_in_top_q}, the top-\(q\) selected samples maintain consistently high proportions of true poisoned pairs across different attacks. 
At the default setting of $q=100$, the purity remains above 90\% across all attacks, ranging from 92.0\% to 98.0\%.
These results demonstrate that \(\mathcal{D}_p\) closely approximates the true poisoned distribution, thereby supporting the exchangeability assumption and enabling valid statistical guarantees in our fine-grained detection.

\subsection{Fault Tolerance Analysis}
\label{appendix:fault_tolerance}
We also investigate the fault tolerance of \alg by varying the quality of the poisoned subset $\mathcal{D}_p$ from the coarse-grained detection stage.
To simulate varying levels of error, we manually construct $\mathcal{D}_p$ by mixing in different proportions of benign pairs. 
Specifically, we consider $\mathcal{D}_p$ sizes of $q = 50$ and $q = 400$, and evaluate detection performance under fault ratios of 20\%, 40\%, 60\%, and 80\% (i.e., the proportion of benign samples mistakenly included in $\mathcal{D}_p$). 
Tab.~\ref{tab:fault_tolerance} reveals a gradual decline in AUROC as the proportion of benign pairs in $\mathcal{D}_p$ increases. 
However, even under extreme conditions where 80\% of $\mathcal{D}_p$ consists of benign samples (indicating a severely degraded first-stage detection), AUROC remains above 0.87 for $q = 50$ and above 0.91 for $q = 400$. 
The above results demonstrate the strong fault tolerance and robustness of \alg even when the initial poisoned subset selection is imperfect.

\begin{figure}[t]
\centering
\includegraphics[width=\linewidth]{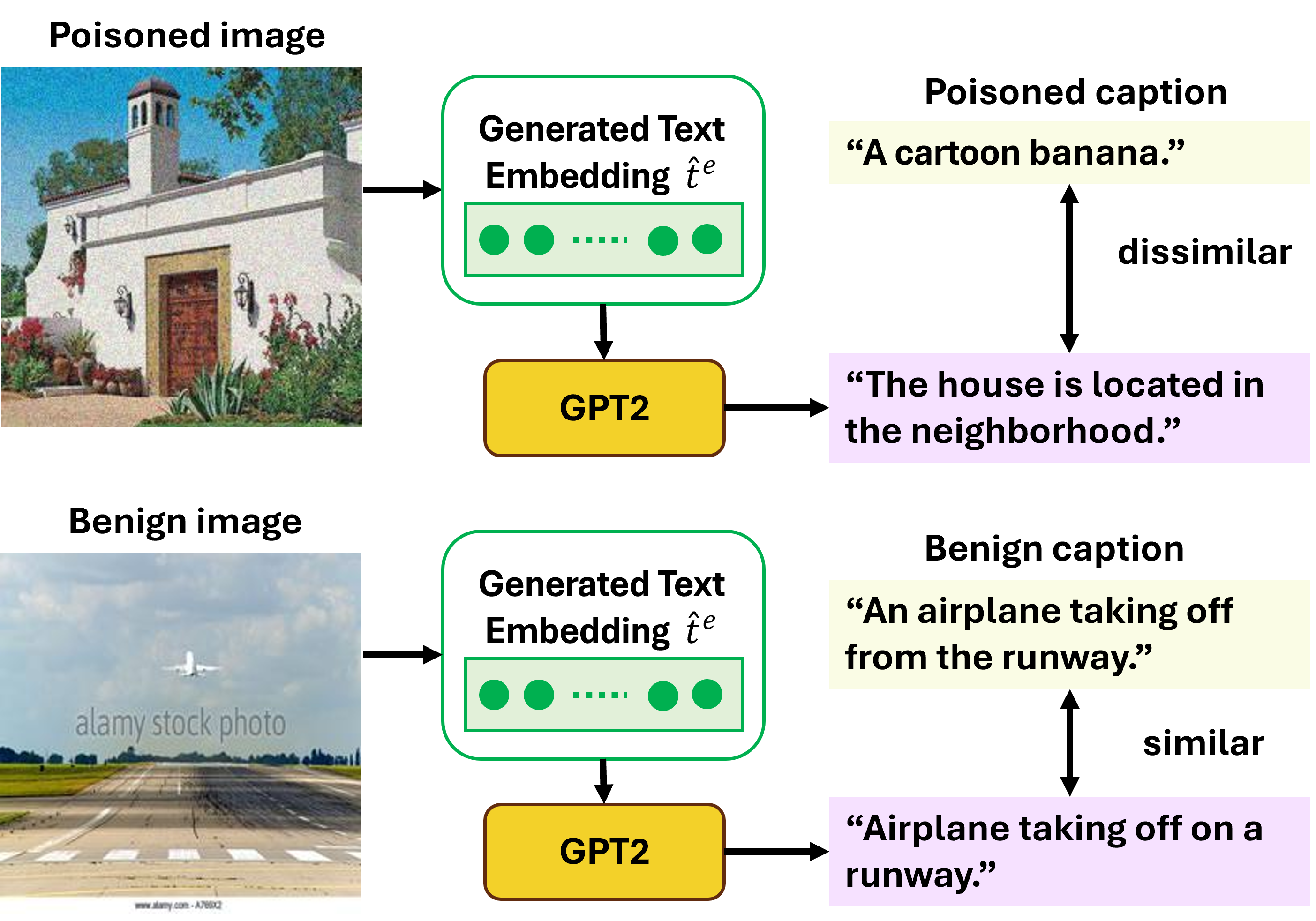}
\caption{Visualization examples of the visual-guided text embedding generated in the coarse-grained stage of \alg.}
\label{fig:visual}
\end{figure}

\subsection{Visualization of Visual-Guided Text Embedding}
\label{app:visualization}
To illustrate the effectiveness of the proposed visual-guided text embedding, we use the pre-trained GPT-2 \cite{radford2019language} to generate textual descriptions from it. 
As shown in Fig.~\ref{fig:visual}, the results reveal clear semantic differences: for poisoned images, the generated captions exhibit significantly lower similarity to their poisoned counterparts, whereas for benign images, the generated captions maintain high consistency with the original descriptions. 
This observation not only validates the robustness of our mapping network but also supports using consistency as a key discriminative criterion in the framework.


\section{Limitations}
Despite its merits, \alg does have limitations. Specifically, it cannot counter attacks that require dynamically adjusting poisoned samples during training. Besides, \alg cannot be immediately generalized to continual learning with CLIP models where new data is continuously added into the training data. \alg cannot be readily extended to the detection at the inference stage. Another issue arises in multilingual settings. If there are multilingual pairs in the image-caption training data, it may be necessary to additionally train a multilingual mapping function $G$ to adapt to the current algorithm design.

\end{document}